\documentclass[sigconf]{acmart} 

\usepackage{array}
\usepackage{flushend}

\usepackage{tikz}
\usetikzlibrary{arrows}

\usepackage{multirow}

\usepackage[titlenumbered,ruled,noend]{algorithm2e}
\DontPrintSemicolon 

\newcolumntype{L}{>{\centering\arraybackslash}m{1cm}}
\newcolumntype{P}[1]{>{\centering\arraybackslash}p{#1}}

\DeclareMathOperator*{\argmax}{\arg\!\max}

\newcommand{\squishlist}{
\begin{list}{$\bullet$}
	{ \setlength{\itemsep}{0pt}      \setlength{\parsep}{3pt}
		\setlength{\topsep}{3pt}       \setlength{\partopsep}{0pt}
		\setlength{\leftmargin}{1.5em} \setlength{\labelwidth}{1em}
		\setlength{\labelsep}{0.5em} } }

\newcommand{\squishlisttwo}{
	\begin{list}{$\bullet$}
		{ \setlength{\itemsep}{0pt}    \setlength{\parsep}{0pt}
			\setlength{\topsep}{0pt}     \setlength{\partopsep}{0pt}
			\setlength{\leftmargin}{2em} \setlength{\labelwidth}{1.5em}
			\setlength{\labelsep}{0.5em} } }
	\newcommand{\squishend}{
\end{list}  }

\def\script#1{\mathcal{#1}}

\def\mS{\script{S}}
\def\mL{\script{L}}

\def\mR{\script{R}}

\def\mI{\script{I}}

\def\mL{\script{L}}

\def\mD{\script{D}}

\def\mH{\script{H}}

{\hspace*{\fill}$\Box$\par}

\begin{document}
	
	
	\copyrightyear{2017} 
	\acmYear{2017} 
	\setcopyright{none}
	\acmConference{Technical Report}{June 2017}{Corvallis, OR}

	\title{Schema Independent Relational Learning}
	
	\author{Jose Picado}
	\affiliation{%
		\institution{Oregon State University}
	}
	\email{picadolj@oregonstate.edu}
	
	\author{Arash Termehchy}
	\affiliation{%
		\institution{Oregon State University}
	}
	\email{termehca@oregonstate.edu}
	
	\author{Alan Fern}
	\affiliation{%
		\institution{Oregon State University}
	}
	\email{alan.fern@oregonstate.edu}
	
	\author{Parisa Ataei}
	\affiliation{%
		\institution{Oregon State University}
	}
	\email{ataeip@oregonstate.edu}
	
	\date{}

	\begin{abstract}
	Learning novel concepts and relations
	from relational databases is an important problem
	with many applications in database systems and machine learning.
	Relational learning algorithms learn the definition of a new relation in terms of existing relations in the database.
	Nevertheless, 
	the same data set may be represented under
	different schemas for various reasons, such as efficiency, data quality, and usability.
	Unfortunately, the output of current relational learning algorithms tends to vary quite substantially over the choice of schema, both in terms of learning accuracy and efficiency. This variation complicates their off-the-shelf application.
	In this paper, we introduce and formalize the property of schema independence of relational learning algorithms, and study
	both the theoretical and empirical dependence of existing algorithms on the common class of (de) composition schema transformations. 
	We study both sample-based learning algorithms, which learn from sets of labeled examples, and query-based algorithms, which learn by asking queries to an oracle.
	We prove that current relational learning algorithms are generally not schema independent.
	For query-based learning algorithms we show that the (de) composition transformations influence their query complexity.
	We propose Castor, a sample-based relational learning algorithm that achieves schema independence by leveraging data dependencies.
	We support the theoretical results with an empirical study that demonstrates the schema dependence/independence of
	several algorithms on existing benchmark and real-world datasets under (de) compositions.
	\end{abstract}

	\maketitle
	
	\section{Introduction}
\label{sec:introduction}

Over the last decade, users' information needs over
relational databases have expanded from
answering precise queries
to using machine learning 
in order to discover
interesting and novel relations and concepts~\cite{Anderson:CIDR:2013,Hellerstein:2012:PVLDB,Cate:2013:LSM:2539032.2539035}.
For instance, consider the UW-CSE database~\cite{Richardson:2006:MLN}, 
which contains information about an academic department.
Given this database, we may want to predict the {\it advisedBy(stud,prof)} relation, which indicates that student {\it stud} is advised by professor {\it prof}.
Machine learning algorithms often assume that data is represented in a single table. 
The contents of the table are the features that capture the essential information required to predict the target relation, i.e., advisedBy. 
In a typical scenario, we would be required to hand-engineer this fixed set of features~\cite{Anderson:CIDR:2013}. 
Each feature would be the result of a query to the database. We would then compute the features for each example in the training data, and store the resulting feature vectors in the table. Finally, we would run a learning algorithm. 

Three challenges arise with the described approach.
First, hand-engineering features is not an easy task. 
It is a slow and tedious process and requires significant expertise~\cite{Anderson:CIDR:2013}. 
It also restricts the algorithm from 
identifying patterns that are not 
reflected in the features or combinations of features. 
Second, by condensing information into a vector of features, we may lose the 
relational structure, which translates into loss of information.
Third, the result of the algorithm may be hard to interpret by users.

In contrast to ``table-based approaches'', relational machine learning (also called relational learning) attempts to learn concepts directly from a relational database. 
Given a database and training instances of a new target relation, relational learning algorithms attempt to induce (approximate) relational definitions of the target in terms of existing relations~\cite{progol,Quinlan:FOIL,QuickFOIL}. 
For example, given an instance of the Original schema for the UW-CSE database in Table~\ref{table:uwcse}, the goal may be to induce a definition of the missing relation {\it advisedBy(stud,prof)} based on a training set of known student-advisor pairs. 
Learned definitions are usually first-order formulas, sometimes restricted to Datalog programs. 
Importantly, such relational learning algorithms do not require the intermediate step of feature engineering. This fact, arguably, allows for the easier deployment of machine learning in the context of relational databases~ \cite{Quinlan:FOIL,QuickFOIL}.

Since the space of possible definitions (e.g. all Datalog rules) is enormous,
relational learning algorithms must employ heuristics, or biases, to search
for effective definitions. 
Unfortunately, such heuristics typically depend on the precise choice of schema of the underlying database, which means that the learning output is schema dependent. This is true even 
if the schemas represent essentially the same information.
As an example, Table~\ref{table:uwcse} shows two schemas for the
UW-CSE database, which is used as a common relational learning
benchmark. 
The original schema was designed by relational learning experts. 
This design is generally discouraged in the database community, as it delivers poor usability and performance in query processing without providing any advantages in terms of data quality in return~\cite{AliceBook}.
A database designer may use a schema closer to
the 4NF schema in Table~\ref{table:uwcse}. 
Because each student {\it stud} has only one {\it phase} and {\it years}, 
a database designer may compose relations {\it student}, {\it inPhase}, 
and {\it yearsInProgram}. She may also combine relations 
{\it professor} and {\it hasPosition}.
This would result in a more understandable schema  
with shorter query execution time, without introducing any redundancy.

\begin{example}
	\label{example:foil_uwcse_definitions}
	We use the classic relational learning algorithm
	FOIL~\cite{Quinlan:FOIL} to induce a definition for {\it advisedBy(stud, prof)} 
	over the Original and 4NF schemas of the UW-CSE database, shown in Table~\ref{table:uwcse}.
	FOIL learns a Datalog rule by starting from an empty rule
	and iteratively adding atoms 
	to the rule such that the resulting rule 
	at each step has the best score: 
	it covers the most positive and the fewest 
	negative examples.
	FOIL learns the following Datalog rule over the 
	UW-CSE database with the Original schema:
	{\small
		\begin{align*}
		\mathit{advisedBy}(x,y) \leftarrow & \mathit{yearsInProgram}(x,7), \mathit{publication}(z,x), \\ & \mathit{publication}(z,y),
		\end{align*} 
	}
	which covers 5 positive examples and 0 negative examples.
	Because $\mathit{yearsInProgram}(x,7)$ covers the most positive and 
	the fewest negative examples in this database, 
	FOIL picks it as the first atom and proceeds by 
	adding the rest of atoms to the rule.
	On the other hand, FOIL learns the following Datalog rule 
	over the 
	4NF schema:
	{\small
		\begin{align*}
		\mathit{advisedBy}(x,y) \leftarrow & \mathit{student}(x,\mathit{post\_generals}, 5), \mathit{professor}(y,\mathit{faculty}), \\
		& \mathit{publication}(z,y), \mathit{taughtBy}(v,y,w),
		\end{align*} 
	}
	which covers 12 positive examples and 10 negative examples.
	In this case, FOIL first selects $\mathit{student}(x,{\mathit post\_generals}, 5)$ because it covers the most 
	positive and the fewest negative examples.
	Because FOIL does not backtrack, then the definitions over both schemas are different, even if the rest of the atoms added to the rules are the same.
	Intuitively, the definition learned over the original schema better expresses the relationship between an advisor and advisee.
\end{example}

\begin{table}
	\small
	\centering
	\begin{tabular} { l l}
		\hline
		Original Schema & 4NF Schema\\
		\hline
		student(\underline{stud}) & student(\underline{stud},phase,years) \\
		inPhase(\underline{stud},phase) &  professor(\underline{prof},position)\\
		yearsInProgram(\underline{stud},years) & publication(\underline{title},\underline{person}) \\
		professor(\underline{prof}) & courseLevel(\underline{crs},level) \\
		hasPosition(\underline{prof},position) & taughtBy(\underline{crs},\underline{prof},\underline{term}) \\
		publication(\underline{title},\underline{person}) & ta(\underline{crs},\underline{stud},\underline{term}) \\
		courseLevel(\underline{crs},level) & \\
		taughtBy(\underline{crs},\underline{prof},\underline{term}) & \\
		ta(\underline{crs},\underline{stud},\underline{term}) &   \\
	\end{tabular}
	\caption{ {\small Schemas for the UW-CSE dataset.} }
	\label{table:uwcse}
\end{table}

Generally, there is no
canonical schema for a particular set of content in practice
and people often represent the same information in
different schemas for several reasons~\cite{AliceBook,infopreserve:XML}.
For example, it is generally easier to enforce
integrity constraints
over highly normalized schemas \cite{AliceBook}.
On the other hand, because more normalized
schemas usually contain many relations,
they are hard to understand and maintain. 
It also takes a relatively long time
to answer queries over database instances
with such schemas \cite{AliceBook}.
Thus, a database designer may sacrifice data quality and
choose a more {\it denormalized}
schema for its data to achieve better usability
and/or performance.
She may also hit a middle ground
by choosing a style of design for some relations
and another style for other relations in the schema.
Further, as the relative priorities of these objectives
change over time, the schema will also
evolve.

Users generally have to restructure their databases, in order to
effectively use relational learning algorithms,
i.e., deliver
definitions for the target concepts that a domain expert
would judge as correct and relevant.
To make matters worse, these algorithms do not normally
offer any clear description of their desired schema and
database users have to rely on their own
expertise and/or do trial and error to find such schemas.
Nevertheless, we ideally want our database analytics algorithms
to be used by ordinary users, not just experts
who know the internals of these algorithms.
Further, the structure
of large-scale databases constantly evolves, 
and we want to move away from the need for
constant expert attention to keep learning
algorithms effective.
Researchers often use
(statistical) relational learning algorithms to solve
various important core database problems,
such as query processing~\cite{Abouzied:PODS:13}, schema mapping~\cite{Cate:2013:LSM:2539032.2539035}, and entity resolution \cite{Getoor:KDD:13}.
Thus, the issue of schema dependence appears in other areas of database management.

One approach to solving the problem of schema dependence
is to run a learning algorithm over {\it all possible schemas}
for a validation subset of the data and select the
schema with the most accurate answers.
Nonetheless, computing all possible schemas
of a DB is generally undecidable \cite{infopreserve:XML}.
One may limit the search space
to a particular family of schemas to
make their computation decidable.
For instance, she may choose to check only schemas that
can be transformed via join and project operations, i.e.
composition and decomposition \cite{AliceBook}.
However, the number of possible schemas
within a particular family of a data set
are extremely large.
For example, a relational table may have
exponential number of distinct decompositions.
As many learning algorithms need some time for
parameter tunning under
a new schema \cite{MLBase:CIDR}, it may take a
prohibitively long time to find the best schema.
Since many relational learning algorithms need to access
the content of the database,
one has to transform the underlying data
to the desired schema, which may not be practical
for a large and/or constantly evolving database.

In this paper, we introduce the novel property of {\it schema independence} for relational learning algorithms, i.e., the ability to deliver the same
answers regardless of the choices of schema for the same data. We propose a formal framework to 
evaluate the property of
schema independence of a relational learning
algorithm for a given family of schema changes. 
Since none of the current relational learning algorithms are schema 
independent, we leverage concepts from database literature to 
design a schema independent algorithm.
To the best of our knowledge,
the property of schema independence has not been
introduced or explored for relational learning algorithms.
The main contributions of this paper are: 
\squishlist
\item
We introduce and formally define the property of
schema independence (Section~\ref{section:framework}), which formalizes the notion of a 
learning algorithm returning equivalent answers over 
schema transformations that preserve
information content.

\item
We analyze the property of schema independence
for the popular family of top-down relational learning algorithms~\cite{progol,Quinlan:FOIL,Richardson:2006:MLN} (Section~\ref{section:top-down}). 
We show that this family of algorithms is not schema independent
under (de) composition transformations.

\item
We explore the property of schema independence for
another family of widely used
relational learning algorithms, known as bottom-up algorithms (Section~\ref{section:bottom-up}).
We formally analyze the bottom-up algorithms Golem~\cite{golem} and ProGolem~\cite{progolem} and show that they are not schema dependent under (de) composition.

\item
We introduce Castor, a bottom-up algorithm that is
provably schema independent under vertical (de) composition (Section~\ref{section:castor}).
Castor achieves schema independence by integrating database constraints, specifically inclusion dependencies, into the learning algorithm.
Castor uses various techniques, such as employing a main-memory RDBMS, to learn efficiently over large databases.

\item Some relational learning algorithms learn the target
concepts using by asking queries from
an {\it oracle}, e.g., a database user, instead of 
using a set of prepared training data \cite{Khardon:1999,Selman:2011,Abouzied:PODS:13}.
They are called 
{\it query-based learning algorithms}.
We formalize the notion of schema independence for 
these type of algorithms and prove that well-known 
algorithms in this category are not schema independent
(Section~\ref{section:worstcase}).

\item
We empirically compare the schema
independence, effectiveness, and efficiency of Castor to some popular relational learning
algorithms under 
(de) composition using a widely used benchmark and
three real-world databases (Section~\ref{section:empirical}).
Our empirical results generally confirm our
theoretical results.
Our results show that Castor is more efficient and 
as effective as or more effective than current algorithms.
\end{list}

	\section{Background}
\label{section:background}

\subsection{Related Work}
\label{section:related}
There has been a growing interest in developing relational learning algorithms that scale to large databases~\cite{QuickFOIL,join:crossmine,amie-plus}.
QuickFOIL~\cite{QuickFOIL} provides an in-RDBMS implementation of a modified version of FOIL.
AMIE+~\cite{amie-plus} learns rules from RDF-style knowledge bases, which contain binary relations. 
These systems focus on scaling learning algorithms for large databases.
We develop a schema independent relational learning algorithm and, as opposed to them, do not modify the internals of the RDBMS.
The system in~\cite{Kumar:2015:LGL:2723372.2723713} 
learns linear models over multiple relations efficiently. 
Our aim, however, is to achieve schema independence.
Also, their system assumes that all relations can be joined into one universal relation, 
which is not generally true in relational databases. 
Moreover, it learns linear models and not Datalog definitions.

Researchers have noticed that using likelihood functions to measure joint probability distributions of attributes 
in a relational database may lead to different results over certain variations of the database design~\cite{Schulte2011}. 
They have proposed using pseudolikelihood functions to approximate the joint probability distributions, which is robust
over some schema variations. Nevertheless, the authors in \cite{Schulte2011} do not provide any general and formal 
framework to explore the sensitivity of learning algorithms to the database design. 

We build upon the body of work on transforming databases without 
modifying their content by
exploring the sensitivity of relational learning algorithms to
such transformations \cite{infopreserve:hull,infopreserve:XML}. 
Another notable group of database transformations is schema mapping for data exchange \cite{DBLP:conf/icdt/FaginKMP03}. 
These transformations may lose information and introduce incomplete information to a database. 
However, for the property of schema independence, a transformation should preserve the information content of databases.
Fagin explores invertible schema mappings that preserve the information content of database instances \cite{Fagin-inverse}.
Nevertheless, these mappings may introduce labeled nulls to the database instance.
To the best of our knowledge, relational learning over database instances with labeled nulls 
has not been precisely defined and explored.
Similar to answering queries over databases with labeled nulls, 
we believe that it is challenging to define reasonable semantics and design efficient and effective 
algorithms for learning relations over such databases. 
Hence, it takes more space than a single paper to include investigations of relational learning over instances with labeled null 
and define schema independence property for transformations that introduce labeled nulls to a database.
We leave them for future work.
Researchers have defined the property of design independence
for keyword query processing over XML \cite{SchemaIndep:ICDE}.
We extend this line of work by formally exploring
the property of schema independence for relational
learning algorithms. We focus on supervised learning algorithms and their schema independence properties over the relational data model.

The architects of the relational model have argued for logical data independence, which oversimplifying a bit, means that an exact query should return the same answers no matter which logical schema is chosen for the data \cite{AliceBook,Codd:Computerworld:85}.
In this paper, we extend the principle of logical data independence 
for relational learning algorithms. 
The property of schema independence also differs with the idea of logical data independence in a subtle but important issue.
One may achieve logical data independence
by an affordable amount of experts' intervention, e.g.,
defining views over the database. 
However, it takes deeper
expertise to find the proper schema for a
learning algorithm, particularly for database applications that contain more than a single learning algorithm. 
Hence, it is less likely to achieve schema independence via expert's intervention.

\subsection{Basic Definitions}
\label{section:basics}
We fix two disjoint (countably) infinite sets of relation 
and attribute symbols. Each relation symbol $R$ is 
associated with a set of attribute symbols denoted 
as $sort(R)$. Let $D$ be a countably infinite domain of
values, i.e., constants.
An instance $I_R$ of relation symbol $R$
with $n=$ $|sort(R)|$ is a (finite) relation over $D^n$.
{\it Schema} $\mR$ is a 
pair $({\bf R}, \Sigma)$, where ${\bf R}$ and $\Sigma$
are finite sets of relations symbols and 
{\it constraints}, respectively.
A constraint restricts the properties of 
data stored in a database.
Examples of constraints are {\it functional dependencies} (FD)
and {\it inclusion dependencies} (IND), i.e., referential
integrity. 
Let $\pi_{X}(I_R)$, $X \subseteq sort(R)$,
denote the projection of relation
$I_R$ on attribute set $X$.
Relation $I_R$ satisfies 
FD $X \rightarrow Y$, where $X, Y$ $\subset sort(R)$,
if for each pair $s,t$ of tuples in $I_R$, 
$\pi_{X}(s)=$ $\pi_{X}(t)$ implies 
$\pi_{Y}(s)=$ $\pi_{Y}(t)$.
Given relation symbols $R$ and $S$ and 
sets of attributes $X \in sort(R)$ and $Y \in sort(S)$, 
relations $I_R$ and $I_S$ satisfy IND 
$R[X] \subseteq S[Y]$ if 
$\pi_{X}(I_R) \subseteq$ $\pi_{Y}(I_S)$. 
If both INDs $R[X]\subseteq$ $S[X]$ and
$S[X]\subseteq$ $R[X]$ hold in a schema, 
we denote them as $R[X] =$ $S[X]$ and call it 
an {\it IND with equality}.
An {\it instance} of schema $\mR$ is a mapping $I$ over $\mR$
that associates each relation $R \in \mR$ to an instance $I_R$ 
that satisfies all constraints in $\Sigma$. 
The set $\Sigma$ may logically imply other constraints, e.g., 
FD $X \rightarrow Y$ and $Y \rightarrow Z$
imply $X \rightarrow Z$ \cite{AliceBook}. The set of 
all constraints implied by $\Sigma$ is shown 
as $\Sigma^{+}$. To simplify our notations,
we use $\Sigma$ and $\Sigma^{+}$ interchangeably.

An {\it atom} is a formula in the form of $R(u_1, \ldots, u_n)$ 
where $R$ is a relation symbol, $n=$ $|sort(R)|$, 
and each $u_i$, $1 \leq i \leq n$, 
is a variable or constant.
If all $u_i$s are constants, the atom is a {\it ground atom}.
A {\it literal} is an atom, or the negation of an atom.
A {\it ground literal} is a literal whose atom is a ground atom.
A {\it definite Horn clause} (Horn clause for short) is a finite set of literals that contains exactly one positive literal. The positive literal is called the head of the clause, and the set of negative literals is called the body.
A clause has the form:
$T({\bf u}) \leftarrow L_1({\bf u_1}), \cdots, L_{n}({\bf u_{n}}).$
Horn clauses are also called conjunctive queries \cite{AliceBook}.
A {\it Horn definition}, i.e., union of conjunctive queries, 
is a set of Horn clauses with the same head literal.
A Horn definition is defined over a schema if the body of all clauses 
in the definition contain only literals whose relations are 
in the schema.
In this paper, we use Horn definitions to define new target relations
that are not in the schema. Thus, the heads of all clauses in these
definitions are the {\it target relation}.

	\section{Framework}
\label{section:framework}

\subsection{Relational Learning}
\label{sec:relational-learning}
Relational learning can be viewed as a search
problem for a hypothesis that deduces the training data,
following either a top-down or bottom-up approach.
Top-down algorithms~\cite{Quinlan:FOIL,progol} start from the most
general hypothesis and employ specialization operators to get more
specific hypotheses. A common specialization operator is the addition 
of a new literal to the body of a clause.
On the other hand, bottom-up algorithms~\cite{golem, progolem}
start from specific hypotheses that are constructed based on ground
training examples, and use generalization operators to search the 
hypothesis space. Generalization operators include inverse 
resolution, relative least general generalization, asymmetric 
relative minimal generalization, among others.
Thus, a relational learning algorithm is a 
sequence of steps, where in each step an operator is applied to the 
current hypothesis.

Inductive Logic Programming (ILP) is the subfield of machine learning that performs relational learning by learning 
first-order definitions from examples and an input relational database.
In this paper we use the names ILP algorithm and relational learning 
algorithm interchangeably.
Training examples $E$ are usually 
tuples of a single target relation 
$T$, which express positive ($E^+$) or negative 
($E^-$) examples. 
The learned definitions are called the hypothesis $H$, 
which is usually restricted to Horn definitions for efficiency reasons. 
A relational learning algorithm takes as input training data $E$, 
database instance $I$, and target relation $T$, and 
learns a hypothesis $H$ that, together with $I$, entails $E$. 
The input database instance $I$ is also 
called {\it background knowledge}.
More formally, the learning problem is described as follows:
\begin{definition}
	Given background knowledge $I$, positive examples $E^+$, negative examples $E^-$, 
	and a target relation $T$, the ILP task is to find a definition $H$ for $T$ such that:
	\squishlist
		\item $\forall p \in E^+, H \wedge I \models p$ (completeness)
		\item $\forall p \in E^-, H \wedge I \not\models p$ (consistency)
	\end{list}
\end{definition}
In the following sections we provide concrete definitions of several relational learning algorithms.

\begin{example}
	\label{example:relational_learning}
	Consider using a relational learning algorithm and the UW-CSE 
	database with the Original schema shown in Table~\ref{table:uwcse} to 
	learn a definition for the target relation 
	$\mathit{collaborated}(x,y)$, 
	which indicates that person $x$ has collaborated with person $y$.
	The algorithm may return definition 
	\begin{align*}
		\small
		\mathit{collaborated}(x,y) \leftarrow \mathit{publication}(p,x), \mathit{publication}(p,y).
	\end{align*} 
	This is a complete and consistent definition with respect 
	to the training data, and indicates that two persons have collaborated if they are co-authors.
\end{example}

In this paper, we study relational learning algorithms for Horn definitions.
We denote the set of all Horn definitions over schema $\mR$ by $\mH\mD_{\mR}$.
This set can be very large, which means that algorithms would need a lot of resources (e.g. time and space) to explore all definitions.
Because, resources are limited in practice, 
algorithms accept parameters that either restrict the hypothesis space or the search strategy.
For instance, an algorithm may consider only clauses whose number
of literals are fewer than a given number, or may follow a greedy approach where only one clause is considered at a time.
Let the {\it parameters} for a learning algorithm be a tuple of variables $\theta = \langle \theta_1, ..., \theta_r \rangle$, where each $\theta_i$ is a parameter for the algorithm.
We denote the parameter space by $\Theta$, and it contains all possible parameters for an algorithm.
We denote the hypothesis space (or language) of algorithm $A$ over
schema $\mR$ with parameters $\theta$ as $\mL^{A}_{\mR, \theta}$.
Note that not all parameters affect the hypothesis space.
For instance, a parameter setting the search strategy to greedy impacts 
how the hypothesis space is explored, but does not restrict the hypothesis space.
The hypothesis space $\mL^{A}_{\mR, \theta}$ is a subset of $\mH\mD_{\mR}$ \cite{progol,Quinlan:FOIL}, and
each member of $\mL^{A}_{\mR, \theta}$ is a hypothesis.

There is a trade-off between computational resources used by 
an algorithm and the size of its hypothesis space.
The hypothesis space is restricted so that the algorithm can be 
used in practice, with the hope that it finds a consistent and complete hypothesis.
\begin{example}
	\label{example:parameter_settings}
	Continuing Example~\ref{example:relational_learning}, consider restricting the hypothesis space to clauses whose number
	of literals are fewer than a given number, which we call clause-length.
	Assume that we are now interested in learning a definition for the target relation {\it collaboratedProf(x,y)}, which indicates that professor $x$ has collaborated with professor $y$, under the Original schema.
	If we set clause-length $= 5$, the learning algorithm is able to learn the complete and consistent definition
    \begin{align*}
    	\small
		\mathit{collaboratedProf}(x,y) \leftarrow& \mathit{professor}(x), \mathit{professor}(y), \\
		& \mathit{publication}(p,x), \mathit{publication}(p,y).
	\end{align*}
	However, if we set clause-length $= 3$, the previous definitions is not in the hypothesis space of the algorithm.
	Thus, the algorithm cannot learn this or any other complete and consistent definition.
\end{example}

\subsection{Schema Independence}
\label{sec:schema_independence}

\subsubsection{Mapping Database Instances}
One may view a schema
as a way of representing background knowledge used
by relational learning algorithms
to learn the definitions of target relations.
Intuitively, in order to learn essentially the same
definitions over schemas $\mR$ and $\mS$, we should
make sure that {\it $\mR$ and $\mS$ represent basically
the same information}.
Let us denote the set of database instances
of schema $\mR$ as $\mI(\mR)$. In order to compare the ability
of $\mR$ and $\mS$ to represent the same information,
we would like to check whether for each database instance
$I \in \mI(\mR)$ there is a database instance
$J \in \mI(\mS)$ that contains basically the same information as
$I$. We adapt the notion of equivalency between schemas
 to precisely state this idea  \cite{infopreserve:hull,infopreserve:XML}.

Given schemas $\mR$ and $\mS$,
a {\it transformation} is a (computable) function
$\tau: \mI(\mR) \rightarrow \mI(\mS)$.
For brevity, we write transformation $\tau$ as
$\tau: \mR \rightarrow \mS$.
Transformation $\tau$ is {\it invertible} iff
it is total and there exists a transformation
$\tau^{-1}: \mS \rightarrow \mR$
such that the composition of $\tau$ and $\tau^{-1}$
is the identity mapping on $\mI(\mR)$, that is
$\tau^{-1}(\tau(I)) = I$ for $I \in \mI(\mR)$.
The transformation $\tau^{-1}$ may or
may not be total.
We call $\tau^{-1}$ the {\it inverse} of $\tau$
and say that $\tau$ is {\it invertible}.
If transformation $\tau$ is invertible,
one can convert every instance $I\in \mI(\mR)$
to an instance $J \in \mI(\mS)$ and reconstruct
$I$ from the available information in $J$.
If $\tau: \mR \rightarrow \mS$ is {\it bijective},
schemas $\mR$ and $\mS$ are 
{\it information equivalent} via $\tau$.
Informally, if two schemas are information equivalent,
one can convert the databases represented
using one of them to the other
without losing any information. Hence,
one can reasonably argue that equivalent schemas
essentially represent the same information.
Our definition of information equivalence between two
schemas is more restricted that the ones proposed in
\cite{infopreserve:hull,infopreserve:XML}.
We assume that in order for schemas $\mR$ and $\mS$ to be information equivalent via $\tau$,
$\tau^{-1}$ has to be total. 
Although more restricted, this definition is sufficient to cover the transformations discussed in this paper.

\begin{example}
\label{example:eqschema}
In addition to the functional dependencies shown in
Table~\ref{table:uwcse},
let the following inclusion
dependencies hold over the relations of Original
schema in this table:
$student[stud]$ $=$ \\$inPhase[stud]$,
$student[stud]$ $=$ $yearsInProgram[stud]$,\\
$professor[prof]$ $=$ $hasPosition[prof]$,
One may join relations $student$, $inPhase$,
and 
$yearsInPrograms$ and join relations
$professor$ and 
$hasPosition$ to map
each instance of the Original schema to
an instance of the 4NF schema.
Also, each instance of the 4NF schema can be
mapped to an instance of the Original schema by
projecting relation $student$ to relations
$student$, $inPhase$, and $yearsInProgram$ and
projecting relation $professor$ to
relations $hasPosition$ and $professor$.
Hence, these schemas are information equivalent.
\end{example}

\subsubsection{Mapping Definitions}
Let $\mH\mD_{\mR}$ be the set of all Horn definitions over schema $\mR$.
In order to learn semantically equivalent definitions over schemas $\mR$ and $\mS$, we should make sure that
the sets $\mH\mD_{\mR}$ and $\mH\mD_{\mS}$ are equivalent. That is, for every definition $h_{\mR} \in \mH\mD_{\mR}$,
there is a semantically equivalent Horn definition in $\mH\mD_{\mS}$, and vice versa.
If the set of Horn definitions
over $\mR$ is a superset or subset of the set of Horn definitions over $\mS$,
it is not reasonable to expect a learning algorithm
to learn semantically equivalent definitions in $\mR$ and $\mS$.

Let $\mL_{\mR}$ be a set of Horn definitions over schema $\mR$ such that $\mL_{\mR} \subseteq \mH\mD_{\mR}$.
Let $h_{\mR} \in \mL_{\mR}$ be a Horn definition over schema $\mR$ and $I \in \mI(\mR)$ be a database instance.
The result of applying a Horn definition $h_{\mR}$ to database instance
$I$ is the set containing the head of all instantiations 
of $h_{\mR}$ for which
the body of the instantiation belongs to $\mI(\mR)$.
$h_{\mR}(I)$ shows the result of $h_{\mR}$ on $I$.

\begin{definition}
\label{def:def-preserving}
Transformation $\tau: \mR \rightarrow \mS$
is {\it definition preserving}
w.r.t. $\mL_{\mR}$ and $\mL_{\mS}$
iff there exists a total function
$\delta_{\tau}: \mL_{\mR} \rightarrow \mL_{\mS}$
such that for every definition $h_{\mR} \in \mL_{\mR}$
and $I \in \mI(\mR)$, $h_{\mR}(I)$
$=\delta_{\tau}(h_{\mR})(\tau(I))$.
\end{definition}
Intuitively, Horn definitions $h_{\mR} $ and
$\delta_{\tau}(h_{\mR} )$ deliver the same
results over all corresponding database
instances in $\mR$ and $\mS$.
We call function $\delta_{\tau}$ a 
{\it definition mapping}
for $\tau$.
Transformation $\tau$ is {\it definition bijective} w.r.t.
 $\mL_{\mR}$ and $\mL_{\mS}$ iff $\tau$ and $\tau^{-1}$ are
definition preserving w.r.t. $\mL_{\mR}$ and $\mL_{\mS}$.

If $\tau$ is definition bijective w.r.t. 
equivalent sets of Horn definitions,
one can rewrite each Horn definition over 
$\mR$ as a Horn definition over $\mS$
such that they return the same results over all corresponding
database instances of $\mR$ and $\mS$, and vice versa.
We call these definitions {\it equivalent}.
We use the operator $\equiv$ to show that two definitions are equivalent.

\subsubsection{Relationship Between Bijective and Definition Bijective Transformations}
In order for a learning algorithm to learn equivalent
definitions over schemas $\mR$ and $\mS$, 
where $\tau: \mR \rightarrow \mS$, $\tau$ should be 
both bijective and definition bijective w.r.t. $\mH\mD_{\mR}$ 
and $\mH\mD_{\mS}$.
If $\tau$ is bijective, the learning algorithm takes as 
input the same background knowledge.
Also, a definition bijective transformation ensures that the 
learning algorithm can output equivalent Horn definitions over 
both schemas.
Nevertheless, it may be hard to check both
conditions for given schemas. Next, we extend the
results in \cite{infopreserve:XML} to
find the relationship between
the properties of
bijective and definition bijective 
transformations.
In this paper, we consider only transformations that can be 
written as sets of Horn definitions.
We call these {\it Horn transformations}.
Composition/ decomposition are well-known examples of
Horn transformations \cite{AliceBook}.
\begin{example}
\label{example:transformations}
Let $\mR$ be the Original schema and $\mS$ be the 4NF schema in Example~\ref{example:eqschema}. The transformation from the Original schema to the 4NF schema 
	can be written as the following set of Horn definitions:
	{\small 
	\begin{align*}
	\mathit{student}(x,y,z) \leftarrow& \mathit{student}(x), \mathit{inPhase}(x,y), \\
	& \mathit{yearsInProgram}(x,z). \\
	\mathit{professor}(x,y) \leftarrow& \mathit{professor}(x), \mathit{hasPosition}(x,y). \\
	\mathit{publication}(x,y) \leftarrow& \mathit{publication}(x,y).
	\end{align*}
	}
	\\
	The inverse of this transformation from the 4NF to Original schema is a set of projection operators, which can also be written as a set of Horn definitions.
\end{example}
Let transformation $\tau: \mR \rightarrow \mS$ and its 
inverse $\tau^{-1}: \mS \rightarrow \mR$
be Horn transformations.
Clearly, the head of each Horn definition in 
$\tau^{-1}$ will be a relation in $\mR$.
Let $h_{\mR}$ be a Horn definition in $\mH\mD_{\mR}$.
The composition of $h_{\mR}$ and $\tau^{-1}$,
denoted by $h_{\mR} \circ \tau^{-1}$,
is a Horn definition that belongs to $\mH\mD_{\mS}$, created
by applying $h_{\mR}$ to the heads of
clauses in $\tau^{-1}$ \cite{AliceBook}.
That is, $h_{\mR} \circ \tau^{-1} (J) = h_{\mR}(\tau^{-1}(J))$, for all $J \in \mI(\mS)$.

\begin{proposition}
\label{proposition:program-equivalent}
Given schemas $\mR$ and $\mS$,
if transformation $\tau: \mR \rightarrow \mS$ is
bijective and both $\tau$ and $\tau^{-1}$
are Horn transformations, then $\tau$ is definition bijective
w.r.t $\mH\mD_{\mR}$ and $\mH\mD_{\mS}$.
\end{proposition}
\begin{proof}
	Let us define a function 
	$\delta_{\tau}: \mH\mD_{\mR} \rightarrow \mH\mD_{\mS}$
	to be $\delta_{\tau}(h_{\mR}) = h_{\mR} \circ \tau^{-1}$ for any $h_{\mR} \in \mH\mD_{\mR}$.
	We know that $\delta_{\tau}(h_{\mR}) \in \mH\mD_{\mS}$.
	Furthermore, for every $h_{\mR} \in \mH\mD_{\mR}$ and $I \in \mI_{\mR}$,
	$h_{\mR}(I) = h_{\mR}(\tau^{-1}(\tau(I))) = (h_{\mR} \circ \tau^{-1})(\tau(I))) = \delta_{\tau}(h_{\mR})(\tau(I))$.
	Similarly, we define a function $\delta^{'}_{\tau}: \mH\mD_{\mS} \rightarrow \mH\mD_{\mR}$ as 
	$\delta^{'}_{\tau}(h_{\mS}) = h_{\mS} \circ \tau$ for any $h_{\mS} \in \mH\mD_{\mS}$.
	Clearly, $\delta^{'}_{\tau}(h_{\mS}) \in \mH\mD_{\mR}$.
	Also, for every $h_{\mS} \in \mH\mD_{\mS}$ and every $J \in \mI_{\mS}$ such that there is an $I \in \mI_{\mR}$ where $I = \tau{J}$,
	$h_{\mS}(J) = h_{\mS}(\tau(I)) = (h_{\mS} \circ \tau)(I) = \delta^{'}_{\tau}(h_{\mS})(I)$.
	Thus, $\tau$ is definition preserving w.r.t. $\mH\mD_{\mR}$ and $\mH\mD_{\mS}$. 
\end{proof}
Intuitively, if 
$\tau: \mR \rightarrow \mS$ is bijective and both 
$\tau$ and $\tau^{-1}$ are 
Horn transformation, every Horn 
definition in $\mH\mD_{\mR}$ can be 
rewritten as a Horn definition in $\mH\mD_{\mS}$ such that they return the 
same results over equivalent database instances.
Hence, in the rest of this paper, 
we consider only the 
bijective Horn transformations whose 
inverses are Horn transformations.

\begin{example}
\label{example:eqlang}
Let $\mR$ be the Original schema and $\mS$ be the 4NF schema in Example~\ref{example:eqschema} and $\tau: \mR \rightarrow \mS$ 
$\tau^{-1}: \mS \rightarrow \mR$ are the Horn transformation explained in Example~\ref{example:transformations}.
According to Proposition \ref{proposition:program-equivalent},
$\tau$ is definition bijective w.r.t. $\mH\mD_{\mR}$ and $\mH\mD_{\mS}$.
\end{example}

\subsubsection{Schema Independence Property}
\label{sec:schema-independence-property}
The {\bf hypothesis space} determines the set of possible Horn definitions that the algorithm can explore.
Therefore, the output of a learning algorithm depends on its hypothesis space.
In Example \ref{example:parameter_settings}, we showed that an algorithm is able to learn a definition for a target relation with some hypothesis space but not in another more restricted space. 
In order for an algorithm to
learn semantically equivalent definitions for a
target relation over schemas $\mR$ and $\mS$, it should have
equivalent hypothesis spaces over $\mR$ and $\mS$.
We call this property hypothesis invariance.
Let $\Theta$ be the parameter space for algorithm $A$.

\begin{definition}
	\label{hypothesis_invariace}
	Algorithm $A$ is hypothesis invariant under transformation $\tau: \mR \rightarrow \mS$ iff
	$\tau$ is definition bijective w.r.t. 
	$\mL^{A}_{\mR, \theta}$ and $\mL^{A}_{\mS, \theta}$, for all $\theta \in \Theta$.
\end{definition}
Algorithm $A$ is hypothesis invariant under 
a set of transformations iff $A$ is hypothesis invariant 
under every transformation in the set.
We now define the notion of schema independence for relational learning algorithms over a bijective transformation.
We define a relational learning algorithm as a function
$A(I, E, \theta)$ to $\mL^{A}_{\mR, \theta}$. 
That is, taking as input a database instance $I$,
training examples $E$, and parameters $\theta \in \Theta$, the
algorithm outputs a hypothesis in
$\mL^{A}_{\mR, \theta}$.
\begin{definition}
	\label{schema_independence}
	 Algorithm $A$ is schema independent
	 under bijective transformation 
	 $\tau:\mR \rightarrow \mS$ iff $A$ is hypothesis 
	 invariant under $\tau$ and
     for every $I \in \mI(\mR)$ 
     and all $\theta \in \Theta$,
	 we have:
$A(\tau(I), E, \theta) \equiv \delta_{\tau} (A(I, E, \theta))$,
		where $\delta_{\tau}$ is the definition
		mapping for $\tau$.
\end{definition}
Algorithm $A$ is schema independent under the 
set of transformations iff it is schema independent under 
each transformation in the set.
Note that if an algorithm is schema independent under transformation $\tau$, it is hypothesis invariant under $\tau$.
However, it is possible for an algorithm not to be schema independent, but be hypothesis invariant. In such cases, the cause of schema dependence must necessarily be related to the search process of the algorithm, rather than hypothesis representation capacity.

\begin{example}
	\label{example:schema_independence}
	Consider the Original schema and the 4NF schema in in Example~\ref{example:eqschema}. The Original schema 
	is the result of a decomposition of the 4NF schema.
	Consider the learning algorithm FOIL. 
	If the target relation is {\it collaboratedProf(x,y)}, as in 
	Example~\ref{example:parameter_settings}, FOIL is able to learn 
	equivalent definitions under the Original schema and the 
	4NF schema.
	But, if the target relation is {\it advisedBy(x,y)}, FOIL 
	learns non-equivalent definitions under these schemas, as seen in 
	Example \ref{example:foil_uwcse_definitions}, and 
	is not schema independent.
\end{example}

	\section{Decomposition and Composition}
\label{transformations}
There are many bijective Horn
transformations between relational schemas
\cite{infopreserve:hull,AliceBook}. It takes
more space than a single paper to explore
the behavior of relational learning algorithms over 
all such transformations.
In this paper, we explore the schema independence of
relational learning algorithms under two widely used 
Horn transformations called {\it decomposition}, where the 
transformation is projection, and {\it composition}, where 
the transformation is natural join \cite{AliceBook}.
Our reasons for selecting these transformations are two fold.
First, they are used in most normalizations and
de-normalizations, e.g., 3rd normal
form. which are 
arguably one of the most frequent schema modifications 
and their importances have been recognized from the early days of 
relational model \cite{AliceBook}.
Database designers often normalize schemas to 
remove redundancy and insertion/ deletion anomalies 
and denormalize them to improve 
query processing time and schema readability \cite{AliceBook}. 
We also observe several cases of them
in relational learning benchmarks, one
of which is presented in Section~\ref{sec:introduction}.

We define decomposition as follows \cite{AliceBook}.
Let $S_i \bowtie S_j$ and $I_{S_i} \bowtie I_{S_j}$ denote the natural join between $S_i$ and
$S_j$ and their instances, respectively. 
We restrict the definition of natural join for the cases where
$S_i$ and $S_j$ have at least one attribute symbol in common 
to avoid Cartesian product.
Let $\bowtie_{i=1}^{n} S_i$ show the natural
join between $S_1,$ $\ldots$, $S_n$. 
Recall that if both INDs $S_1[A]\subseteq$ $S_2[B]$ and
$S_2[B]\subseteq$ $S_1[A]$ hold in a schema, 
we denote them as $S_1[A] =$ $S_2[B]$ and call it 
an IND with equality.
\begin{definition}
\label{def:decompositionGeneral}
A decomposition of schema $\mR=$ $({\bf R}, \Sigma_{R})$ 
with single relation symbol $R$ is schema $\mS=$ 
$({\bf S}, \Sigma_{S})$ with relation symbols $S_1\ldots$ $S_n$ 
such that $sort(R)=$ $\cup_{1 \leq i \leq n} sort(S_i)$
and 
\squishlist
\item For each relation $I_R$ there is one and only one instance
$(I_{S_1} \ldots I_{S_n})$ of $\mS$ such that 
$\pi_{sort(S_i)}(I_R)$ $=I_{S_i}$, $1 \leq i \leq n$,
and $\bowtie_{i=1}^{n} I_{S_i}$ $=I_R$.

\item For all $S_i, S_j$, $1 \leq i,j \leq n$, 
such that $X =$ $sort(S_i) \cap sort(S_j)$ $\neq \emptyset$,
$\Sigma_{\mS}$ contains IND with equality 
$S_i[X]=$ $S_j[X]$.

\item We have $\Sigma_S =$ $\Sigma_R \cup \lambda$. 
\end{list}
\end{definition}
The first and third conditions in Definition~\ref{def:decompositionGeneral} 
are generally known as {\it lossless join} and
{\it dependency preservation} properties, respectively.
The second condition in Definition~\ref{def:decompositionGeneral}
ensures that the natural join 
of relations in every instance $I_{\mS}$ of $\mS$ 
does not lose any tuples in $I_{\mS}$.  
Table~\ref{table:uwcse} depicts 
an example of a decomposition. 
Relation symbol $student$ in the 4NF schema is 
decomposed into $student$,
$inPhase$, and $yearsInProgram$ in the original schema.
The conditions of Definition~\ref{def:decompositionGeneral}, 
e.g., lossless join property, hold in this example
due to the FDs in original and 4NF schemas \cite{AliceBook}.
These conditions may also be satisfied because of other types of constraints in the schema, such as multi-valued dependencies.
A {\it composition} is the inverse of a decomposition, which
is expressed by natural join.

Consider again schema $\mS$ in Definition~\ref{def:decompositionGeneral}.
The join $\bowtie_{i=1}^{n} I_{S_i}$ is {\it globally consistent} if for each $j$, $1\leq j \leq n$, 
$\pi_{\mathit{sort(S_j)}}$ $\bowtie_{i=1}^{n} I_{S_i}$ $=I_{S_j}$ \cite{AliceBook}.
Intuitively speaking, a join is globally consistent if none of its relation has a dangling tuple 
regarding the join. For example, the join between the relations of $\mS$ in the first condition of 
Definition~\ref{def:decompositionGeneral} is globally consistent. 
The join $\bowtie_{i=1}^{n} I_{S_i}$ is {\it pairwise consistent} 
if for each $1 \leq i,j \leq n$, $\pi_{\mathit{sort}(S_i)} (I_{S_i} \bowtie I_{S_j})$ $= I_{S_i}$.
In other words, $I_{S_i}$ does not lose any tuple after joining with $I_{S_j}$.
The join $\bowtie_{i=1}^{n} S_i$ is {\it acyclic} if each instance 
$\bowtie_{i=1}^{n} I_{S_i}$ that is pairwise consistent is globally consistent \cite{AliceBook}.
For example, the join $S_1 \bowtie S_2$ in schema $\mS_1:$$\{S_1(A,B),S_2(A,C)\}$ is acyclic. But, the join
$S_3 \bowtie S_4 \bowtie S_5$ in schema $\mS_2:$ $\{S_3(A,B),S_4(B,C),$ $S_5(B,A),\}$ is cyclic.
In this paper, we consider only the decompositions 
where the join in the first condition of Definition~\ref{def:decompositionGeneral} is acyclic \cite{AliceBook}. 
Acyclic joins cover most decompositions in real-world \cite{AliceBook}. 
For examples, most normal forms, e.g., 3NF, BCNF, 4NF, have acyclic joins.

For simplicity, we consider leaving a relation unchanged as 
a special case of decomposition.
We define the decomposition (composition) of a schema with more than
one relation as the set of decompositions (compositions) of all its
relations. We define a
{\it decomposition/ composition} of a schema
as a finite set of applications of composition and/or
decomposition to the schema.
Every decomposition is bijective \cite{AliceBook}.
Because each decomposition is bijective, every composition is also bijective. 
Because both projection and natural join can be written 
as Horn definitions, each decomposition/ composition 
and its inverse are Horn transformations. 
Hence, they are definition bijective.
We explore the property of schema independence only 
for decomposition/ composition in this paper.

	\section{Top-down algorithms}
\label{section:top-down}
Most relational learning algorithms follow a covering approach \cite{Quinlan:FOIL,progol}.
The covering approach consists in constructing one clause at a time.
After building a clause, the algorithm adds the clause to the hypothesis, discards the positive examples covered by the clause,
and moves on to learn a new clause.
Algorithm~\ref{algorithm:generic-covering} sketches a generic relational learning algorithm that follows a covering approach.
The strategy followed by the {\it LearnClause} procedure depends on the nature of the algorithm.
In top-down algorithms, the {\it LearnClause} procedure in
Algorithm~\ref{algorithm:generic-covering}
searches the hypothesis space from general to specific, by using a refinement (specialization) operator that is 
generally adding a new literal to the clause.

\begin{algorithm}
	\SetKwInOut{Input}{Input}
	\SetKwInOut{Output}{Output}
	\Input{Database instance $I$, positive examples $E^+$, negative examples $E^-$}
	\Output{A Horn definition $H$}
	
	$H \leftarrow \{\}$ ; 	$U \leftarrow E^+$\;
	\While{$U$ is not empty}{
		$C \leftarrow LearnClause(I,U,E^-)$\;
		\If{$C$ satisfies minimum condition}{
			$H \leftarrow H \cup C$\;
			$U \leftarrow U - \{ c \in U | I \cup H \models c \}$\;
		}
	}
	return $H$\;
	
	\caption{{\small Generic relational learning algorithm following a covering approach.}}
	\label{algorithm:generic-covering}
\end{algorithm}

The hypothesis space in top-down algorithms is a refinement graph, that is a rooted directed acyclic graph in which nodes represent clauses and each arc is the application of a basic refinement operator.
The basic strategy of top-down algorithms consists of starting from the most general clause, which corresponds to the root of the refinement graph, and repeatedly refining it until it does not cover any negative example.
Figure~\ref{figure:refinementgraph} shows fragments of the
refinement graph for learning the definition of
$collaborated$ relation over the original schema of Table~\ref{table:uwcse}.
Because of space constraints, we do not show the head of
the clause $collaborated$ in any node of the refinement graph
in Figure~\ref{figure:refinementgraph} but its root.

The strategy of constructing and searching the refinement graph
varies between different top-down algorithms.
For instance, FOIL~\cite{Quinlan:FOIL,QuickFOIL} 
is an efficient and popular
top-down algorithm that follows a greedy best-first search
strategy.
In this section, we analyze the schema independence properties of 
FOIL. However, the results that we show in this section hold for all 
top-down algorithms no matter which search strategy they follow.

\begin{figure}
	\centering
	\begin{tikzpicture}[auto]
	\tikzset{edge/.style = {->,> = latex'}}
	\centering
	
	\node[] (1) [] at (0,0) {{\it \small collaborated(x,y) $\leftarrow$ true}};
	\node[] (2) [] at (-3.5,-1) {{\it \small  $\leftarrow$ student(x)}};
	\node[] (3) [] at (-1,-1) {{\it \small  $\leftarrow$ inPhase(x,p)}};
	\node[] (4) [] at (1,-1) {$\cdots$};
	\node[] (5) [] at (3,-1) {{\it \small  $\leftarrow$ publication(p,x)}};
	\node[] (6) [] at (-4,-2) {$\cdots$};
	\node[] (7) [] at (-2,-2) {$\cdots$};
	\node[] (8) [] at (1,-2) {{\it \small  $\leftarrow$ publication(p,x), publication(p,y)}};
	\node[] (9) [] at (4,-2) {$\cdots$};
	
	\draw[edge] (1) to (2);
	\draw[edge] (1) to (3);
	\draw[edge] (1) to (4);
	\draw[edge] (1) to (5);
	\draw[edge] (2) to (6);
	\draw[edge] (3) to (7);
	\draw[edge] (5) to (8);
	\draw[edge] (5) to (9);
	
	\end{tikzpicture}
	\caption{Fragments of a refinement graph for $collaborated$.}
	\label{figure:refinementgraph}
\end{figure}
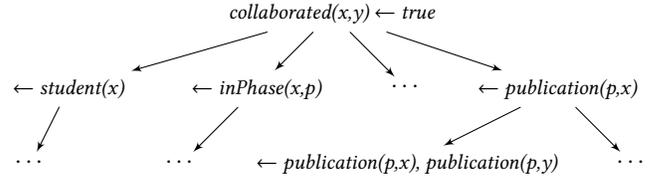

The refinement graph for most schemas, even the ones with a relatively
small number of relations and attributes,
may grow significantly \cite{Quinlan:FOIL,progol}.
Hence, the construction and search over the refinement graph may
become too inefficient to be practical.
To be used in practice, FOIL restricts its search space, i.e. hypothesis space.
We call the number of literals in a clause its length.
A common method is to restrict the maximum length of each clause in the refinement graph \cite{Quinlan:FOIL,progol}.
Intuitively, because composition/decompositions modify
the number of relations in a schema, equivalent clauses over
the original and transformed schemas may have different lengths. Hence, this type of restrictions may result
in different hypothesis spaces. 
One may like to fix this problem by
choosing different values for the maximum lengths over
the original and transformed schemas. 
The following theorem proves that it is not possible to
achieve equivalent hypothesis spaces 
by restricting the maximum length of clauses no matter
what values are used over the original and transformed schemas.
\begin{theorem}
	\label{proposition:topdown-parameters}
	FOIL is not hypothesis invariant.
\end{theorem}
\begin{proof}
	Let relations $R_1(A, B, C)$ and
	$R_2(D, B, E)$ be in $\mR$ and 
	$\tau:\mR \rightarrow \mS$ decompose 
	$R_1$ to $S_1(A, B)$ and $S_2(B, C)$
	and $R_2$ to $S_3(D, B)$ and $S_4(B, E)$.
	Let $l$ be the maximum clause length 
	and $\theta = \langle l \rangle$ be the parameter setting 
	for FOIL. Without loss of generality we set the 
	value of $l$ to 2.
	Let $T(x,y)$ be the target relation. 
	Consider hypothesis $h_{\mR}$:
	$T(x, y) \leftarrow$ $R_1(x, z, w), R_2(y, z, v)$
	over schema $\mR$ whose  
	mapped hypothesis $\delta_{\tau}(h_{\mR})$ is:
	$T(x, y) \leftarrow$ $S_1(x, z)$ $,S_2(z, w),$ $S_3(y, z), S_4(z, v)$.
	Hypothesis $h_{\mR}$ is in the hypothesis language 
	$\mL^{FOIL}_{\mR, 
		\theta}$ because it clause length is equal to 2.
	However, hypothesis $\delta_{\tau}(h_{\mR})$ is not in the hypothesis language $\mL^{FOIL}_{\mS, \theta}$ because its
	clause length exceeds 2.
	Therefore, hypothesis spaces $\mL^{FOIL}_{\mR, \theta}$ and 
	$\mL^{FOIL}_{\mS, \theta}$ are not equivalent.
	To achieve hypothesis equivalence, one may change the parameter 
	setting for $\mS$ to $\theta'$ with $l=4$ so that the hypothesis 
	$\delta_{\tau}(h_{\mR})$ becomes a member of 
	$\mL^{FOIL}_{\mS, \theta'}$. 
	This modification also brings the new hypothesis 
	$T(x, y) \leftarrow$ $S_1(x, z),$ $S_1(x, w),$ $S_1(x, t), S_1(x, 
	y)$ to $\mL^{FOIL}_{\mS, \theta'}$.
	The equivalent hypothesis to this 
	new hypothesis over $\mR$ is
	$T(x, y) \leftarrow$ $R_1(x, z, v_1)$ $,R_1(x, w, v_2),$ 
	$R_1(x, t, v_3),$ $R_1(x, y, v_4)$
	where $v_i$s
	are fresh variables. Because this 
	hypothesis over $\mR$ is minimal, one has to also change $l$ over 
	$\mR$ to 4 to achieve equivalent hypothesis spaces over $\mR$ and 
	$\mS$. Hence, we have to alternate between the parameter settings 
	over $\mR$ and $\mS$ without any stopping condition. Thus, 
	there are not any fixed parameter settings that ensure the 
	hypothesis equivalence over $\mR$ and $\mS$.
\end{proof}

Progol is another popular top-down algorithm that
follows the same approach as FOIL but considers a larger number of candidate clauses at each step over a more restricted hypothesis space \cite{progol}.
Theorem~\ref{proposition:topdown-parameters} also applies to Progol.
It is confirmed by our experiments in Section~\ref{section:empirical}.

	\section{Bottom-up Algorithms}
\label{section:bottom-up}
Bottom-up algorithms also follow the covering approach shown in Algorithm~\ref{algorithm:generic-covering}.
However, their {\it LearnClause} procedure searches the hypothesis space from specific to general hypotheses.
Given a positive example, bottom-up algorithms attempt to find the most specific clause in the hypothesis space, called {\it bottom-clause}, that covers the example, relative to the database instance \cite{golem,progolem}. 
They generalize these bottom-clauses to find
definitions that cover as most positive and as fewest negative examples as possible.

\subsection{Bottom-clause Construction}
\label{section:bottomclause}
Let $I_{\mR}$ be a database instance over schema $\mR$.
The bottom-clause associated with positive example $e$, relative to $I_{\mR}$, denoted by $\perp_{e,I_{\mR}}$,
is the most specific clause over $\mR$ that covers $e$, relative to $I_{\mR}$. 
A typical algorithm for computing bottom-clauses using inverse entailment is given in \cite{progol}.
The algorithm starts with an empty clause, and iteratively adds literals to the clause.
Given positive example $T(a_1, \ldots, a_n)$, it assigns a fresh variable $u_i$ to each distinct 
constant and adds literal $T(u_1, \ldots, u_n)$ to the head of the bottom-clause. 
The algorithm maintains the mapping between constants and variables. 
It then finds all tuples in the database that contain constants $a_1, \ldots, a_n$.
For each tuple, the algorithm adds a new literal to the bottom-clause, 
where the predicate symbol is the tuple relation symbol
and the terms are variables obtained by replacing $a_1, \ldots, a_n$ in the tuple to their corresponding variables
and assigning new variables to newly encountered constants in the tuples.
. 
In the following iterations, the algorithm searches the database
for tuples that contain new constants and adds new literals to the 
bottom-clause. 
This algorithm may generate very long clauses after multiple iterations over a large database.
A common method to restrict the number of iterations is to limit the maximum {\it depth} of the bottom-clause \cite{progol}.
The depth of a variable $x$, denoted by $d(x)$, is 0 if it appears in the head of the clause, 
otherwise it is $min_{v \in U_x}(d(v))+1$, where $U_x$ are the variables of literals in the body of the clause containing $x$. 
The depth of a literal is the maximum depth of the variables appearing in the literal.
The depth of a clause is the maximum depth of the literals appearing in the clause.
The algorithm creates literals of depth at most $i$ in iteration $i$.
\begin{example}
	\label{example:depth}
	This clause over the Original UW-CSE schema in Table~\ref{table:uwcse} 
	has depth 1: {\small $taLevel(x,y) \leftarrow$} {\small $ta(c,x,t), courseLevel(c,y)$}.
	The following clause for target relation $commonLevel(x,y)$, 
	which says that students $x$ and $y$ assist with courses at the same level has depth 2:
	{\small 
	\begin{align*}
		\mathit{commonLevel}(x,y) \leftarrow&  \mathit{ta}(c1,x,t1), \mathit{ta}(c2,y,t2), \\
		& \mathit{courseLevel}(c1,l), \mathit{courseLevel}(c2,l).
	\end{align*}
	}
\end{example}
Bottom-clauses determine the hypothesis space of a bottom-up algorithm: longer bottom-clauses allow the algorithm 
to explore larger number of definitions. 
To be schema independent, bottom-up algorithms must get equivalent bottom-clauses associated with the same example, relative 
to equivalent instances of the original and transformed schemas.
Otherwise, these algorithms will not be hypothesis invariant. 
Using the depth parameter does not result in such equivalent bottom-clauses,  
because the original and transformed schemas need different 
depths to create equivalent bottom-clauses.
\begin{example}
	Let us compose and replace relations 
	{\small $courseLevel$ $(crs,level)$} and 
	{\small $ta(crs, stud, term)$} 
	in the Original UW-CSE schema with
	$courseLevelTa(crs,level,stud,term)$.
	{\small $commonLevel$} from Example \ref{example:depth}
	has the following definition over this schema, which has depth 1:
	{\small $commonLevel(x,y) \leftarrow$} {\small $courseLevelTa(c1,l,x,t1),$}
	{\small $courseLevelTa(c2,l,y,t2)$}.
	If we set the maximum depth to 1, in the Original schema, the clause in Example \ref{example:depth} is not in the hypothesis language.
	But, under the new schema, the clause presented above is in the hypothesis language.
\end{example}
Using a similar idea to the proof of 
Theorem~\ref{proposition:topdown-parameters}, the following
lemma proves that the bottom-clause construction 
algorithm is schema dependent even if different depth 
values are used across schemas.
\begin{lemma}
Bottom-clause construction is schema dependent.
\end{lemma}

\subsection{Generalization}
There are multiple bottom-up algorithms whose 
differences lie mainly in their generalization operator \cite{golem,progolem,Arias:2007}. 
We consider two algorithms that are representative of the family of bottom-up algorithms: Golem~\cite{golem} and ProGolem~\cite{progolem}.

\subsection{Golem}
\label{golem}

In this section, we consider a bottom-up learning algorithm called Golem~\cite{golem}.
Golem, like other learning algorithms, follows a covering approach, as the one shown in Algorithm~\ref{algorithm:generic-covering}.
Golem's {\it LearnClause} procedure follows a bottom-up approach, which is based on the {\it least general generalization} ($lgg$) operator.
Given clauses $C_1$ and $C_2$, the $lgg$ of $C_1$ and $C_2$ is the clause $C$ that is more general than $C_1$ and $C_2$, but the least general such clause.
The notion of generality is defined by $\theta$-subsumption. Therefore, clause $C$ is more general than $C_1$ if and only if $C$ $\theta$-subsumes $C_1$ (and similarly for $C_2$).
This notion of generality gives a computable generality relation. Further, the $lgg$ of two clauses is unique.
Because of the lack of space, for further details we refer the readers to~\cite{golem,plotkin71furtherNote}.


Golem uses a special case of bottom-clause, where all literals of the clause are grounded. 
We call this type of clause the {\it saturation}.
A saturation can be computed using the bottom-clause construction algorithm described above.
Given the saturations for a pair examples, 
the operator that computes the $lgg$ for the pair of saturations is called the {\it relative least general generalization} ($rlgg$).
The $lgg$ of a set of saturations is defined via pairwise operations, that is
\begin{equation*}
lgg(\{ C_1, ..., C_n \}) = lgg(lgg(\{ C_1, ..., C_{n-1} \}), C_n)
\end{equation*}
The order of pairwise $lgg$s does not matter as the $lgg$ operator is commutative and associative.

Given a database instance $I$ and training examples $E^+$ and $E^-$, Golem's {\it LearnClause} procedure learns a clause that covers the most positive and the fewest negative examples as possible.
Algorithm~\ref{algorithm:golem-learnclause} sketches this procedure.
Intuitively, the algorithm first randomly selects a subset $E^+_S$ of positive examples $E^+$.
It then generates candidate clauses by computing the $rlgg$ between every pair of examples in $E^+_S$.
The algorithm considers only candidate clauses that satisfy some minimum condition, e.g., minimum precision of a clause.
It then greedily includes new examples into the generalization to create new candidate clauses.
This algorithm uses the function $Covers(C,E)$, which returns the examples in $E$ covered by clause $C$.
The algorithm stops when no improvement can be done.

\begin{algorithm}
	\SetKwInOut{Input}{Input}
	\SetKwInOut{Output}{Output}
	\Input{Database instance $I$, positive examples $E^+$, negative examples $E^-$, parameter $K$.}
	\Output{A new clause $C^*$.}
	
	$E^+_S \leftarrow K$ randomly selected positive examples from $E^+$\;
	${\bf C} = \{ C = lgg(\perp_{e,I}, \perp_{e',I}) \; | \; e, e' \in E^+_S, C \text{ satifies minimum condition} \}$\;
	
	\While{${\bf C}$ is not empty}{
		$C^* = \argmax_{C \in {\bf C}} \; Score(C,E^+_S,E^-)$\;
		$E^+_S = E^+_S - Covers(C^*, E^+_S)$\;
		${\bf C} = \{ C = lgg(C^*, \perp_{e,I}) \; | \; e \in E^+_S, C \text{ satifies minimum condition} \}$\;
	}
	return $C^*$
	
	\caption{Golem's {\it LearnClause} algorithm.}
	\label{algorithm:golem-learnclause}
\end{algorithm}

\begin{theorem}
	\label{proposition:rlgg-schemaindependent}
	The $rlgg$ operator is schema independent.
\end{theorem}
\begin{proof}
	Let $\tau: \mR \rightarrow \mS$ be a bijective transformation that is a vertical composition/ decomposition between schemas $\mR = ({\bf R}, \Sigma_{\mR})$ and $\mS = ({\bf S}, \Sigma_{\mS})$.
	Let $I$ and $J$ be instances of $\mR$ and $\mS$, respectively, such that $\tau(I) = J$.
	Let $T$ be the target relation, and $e_1 = T(a_1, \cdots, a_l)$ and $e_2 = T(b_1, \cdots, b_l)$ be two positive examples.
	Let $(e_1 \leftarrow I'_1)$ and $(e_2 \leftarrow I'_2)$ be the saturations under schema $\mR$ for $e_1$ and $e_2$, respectively, such that $I'_1, I'_2 \subseteq I$.
	Similarly, let $(e_1 \leftarrow J'_1)$ and $(e_2 \leftarrow J'_2)$ be the saturations under schema $\mS$ for $e_1$ and $e_2$, respectively, such that $J'_1, J'_2 \subseteq J$.
	
	We show that the result of the $rlgg$ operator for examples $e_1$ and $e_2$ is equivalent under schemas $\mR$ and $\mS$.
	That is 
	\begin{align*}
		rlgg_{\mR}(e_1, e_2) &\equiv rlgg_{\mS}(e_1, e_2) \\
		lgg((e_1 \leftarrow I'_1), (e_2 \leftarrow I'_2)) &\equiv lgg((e_1 \leftarrow J'_1), (e_2 \leftarrow J'_2))
	\end{align*}
	We know that $(e_1 \leftarrow I'_1)$ and $(e_2 \leftarrow I'_2))$ are clauses.
	Therefore, $lgg((e_1 \leftarrow I'_1), (e_2 \leftarrow I'_2))$ is the set of pairwise $lgg$ operations of compatible ground atoms in $(e_1 \leftarrow I'_1)$ and $(e_2 \leftarrow I'_2)$.
	Two atoms are compatible if they have the same relation name.
	We show that the $lgg$ of compatible ground atoms under schema $\mR$ delivers equivalent results under schema $\mS$.
	
	Let $R \in {\bf R}$ be a relation in $\mR$ such that $\tau (R) = S_1, \cdots, S_m$, $1 \le m \le |{\bf S}|$.
	Because of Corollary 4.3.2 in \cite{Atzeni:TCS:82}, we know that if $\tau$ is bijective, $\Sigma_{\mS}$ contains inclusion dependencies between the join attributes of $S_1, \cdots, S_m$.
	Let $r_1 = R(a_1, \cdots, a_k)$ and $r_2 = R(a'_1, \cdots, a'_k)$ be two ground atoms in $I$.
	Then, $\tau (r_1) = S_1(t_1), \cdots,$ $S_m(t_m)$ and $\tau (r_2) = S_1(t'_1), \cdots, S_m(t'_m)$ are ground atoms in $J$, where $t_i$ and $t'_i$, $1 \le i \le m$, are tuples.
	Then, the $lgg$ of ground atoms $r_1$ and $r_2$ is defined as
	\begin{align*}
		lgg(r_1, r_2) ={}& R(lgg(a_1, a'_1), \cdots, lgg(a_k, a'_k))
	\end{align*}
	By applying transformation $\tau$, this is equivalent to
	\begin{align*}
		S_1(s_1), S_2(s_2), \cdots, S_m(s_m)
	\end{align*}
	where $s_j$ is a tuple that contains a subset of attributes in $\{ lgg(a_1, a'_1),$ $\cdots,$ $lgg(a_k, a'_k) \}$ for $1 \le j \le m$.
	By definition of the $lgg$ operator, we get
	\begin{align*}
			S_1(s_1), S_2(s_2), \cdots, S_m(s_m) =& lgg(S_1(t_1), S_1(t'_1)), \\
			& \cdots, \\
			& lgg(S_m(t_m), S_m(t'_m))\\
		=& lgg(\tau (r_1), \tau (r_2))
	\end{align*}
\end{proof}

In Section~\ref{section:CastorBottomClause} we show that the bottom-clause construction algorithm can be modified to be schema independent. 
Because the $rlgg$ operator is also schema independent, Golem can achieve schema independence.
However, Golem may generate very large clauses after each application of the $rlgg$ operator.
The reason is that the size of a clause generated by $lgg(C_1, C_2)$, where $C_1$ and $C_2$ are clauses, is bounded by $|C_1| \cdot |C_2|$.
Let $n$ be the number of positive examples to generalize and $m$ be the maximum length of a bottom-clause. Then, the length of the clause generated by $rlgg$ is bounded by $O(m^n)$, i.e., it grows exponentially in the number of positive examples covered.
This results in exponential running time.
Therefore, an algorithm that uses the $rlgg$ operator, such as Golem~\cite{golem}, cannot learn efficiently over small or medium databases without making assumptions that do not hold over most real-world databases~\cite{progolem}.

\subsection{ProGolem}
\label{progolem}

ProGolem is a bottom-up algorithm that can run efficiently over small or medium databases without making generally unrealistic assumptions~\cite{progolem}.
To explore the hypothesis space and generalize clauses efficiently, 
ProGolem assumes that clauses are ordered. 
An {\it ordered clause} is a clause where the order and duplication of literals matter.
If clause $C$ is considered an ordered clause, then it is denoted as $\overrightarrow{C}$.
For instance, clauses $\overrightarrow{C} = T(x) \leftarrow P(x), Q(x)$, $\overrightarrow{D} = T(x) \leftarrow Q(x), P(x)$, and $\overrightarrow{E} = T(x) \leftarrow P(x), P(x), Q(x)$ are all different.

ProGolem uses the {\it asymmetric relative minimal generalization} ({\it armg}) operator to generalize clauses. 
ProGolem's $\mathit{LearnClause}$ procedure first generates the bottom-clause associated with some positive example. 
Then, it performs a beam search to select the best 
clause generated after multiple applications of the {\it armg} operator.
More formally, given clause $\overrightarrow{C}$, ProGolem randomly picks a subset $E^+_S$ of positive examples to generalize $\overrightarrow{C}$.
For each example $e'$ in $E^+_S$, ProGolem uses the {\it armg} operator to generate a candidate clause $\overrightarrow{C'}$, which is more general than $\overrightarrow{C}$ and covers $e'$.
It then selects the highest scoring candidate clauses to keep in the beam and iterates until the clauses cannot be improved.
The beam search requires an evaluation function to score clauses. 
One may select an evaluation function that is 
agnostic of the schema used, such as coverage, which is the 
number of positive examples minus the number 
of negative examples covered by the clause. 

\begin{algorithm}
	\SetKwInOut{Input}{Input}
	\SetKwInOut{Output}{Output}
	\Input{Bottom-clause $\perp_{e,I_{\mR}}$, positive example $e'$.}
	\Output{An {\it ARMG} of $\perp_{e,I_{\mR}}$ that covers $e'$. } 
	
	$\overrightarrow{C}$ is $\perp_{e,I_{\mR}} = T \leftarrow L_1, \cdots, L_n$\;
	\While{there is a blocking atom $L_i$ w.r.t. $e'$ in the body of $\overrightarrow{C}$}{
		Remove $L_i$ from $\overrightarrow{C}$\;
		Remove atoms from $\overrightarrow{C}$ which are not head-connected\;
	}
	Return $\overrightarrow{C}$\;
	
	\caption{{\it ARMG} algorithm.}
	\label{algorithm:armg}
\end{algorithm}

We now explain the {\it armg} operator in detail.
Let $\perp_{e,I_{\mR}}$ be the bottom-clause associated with example $e$, relative to $I_{\mR}$.
Let $\overrightarrow{C} = T \leftarrow L_1, \cdots, L_n$ be the ordered version of $\perp_{e,I_{\mR}}$. 
Let $e'$ be another example.
$L_i$ is a {\it blocking atom} iff $i$ is the least value such that for all substitutions $\theta$ where 
$e' = T\theta$, 
the clause $\overrightarrow{C'}\theta = (T \leftarrow L_1, \cdots, L_i)\theta$ does not cover $e'$, relative to $I_{\mR}$~\cite{progolem}.
Algorithm~\ref{algorithm:armg} shows the {\it ARMG} algorithm, which implements the {\it armg} operator.
Given the bottom-clause $\perp_{e,I_{\mR}}$ and a positive example $e'$, {\it armg} drops all blocking atoms from the body of $\perp_{e,I_{\mR}}$
until $e'$ is covered. 
After removing a blocking atom, some literals in the body may not have any variable in common with the other literals 
in the body and head of the clause, i.e., they are not {\it head-connected}. 
{\it Armg} also drops those literals.
For ProGolem to be schema independent, 
the {\it armg} operator must return equivalent clauses given
equivalent input clauses over original and transformed databases.
\begin{example}
	\label{example:armg-notrobust}
	Consider the following equivalent definitions for target relation $\mathit{hardWorking}$
	over the Original and 
	4NF UW-CSE schema in Table~\ref{table:uwcse},
	respectively:\\
	{\small $\mathit{hardWorking(x)} \leftarrow$} {\small $\mathit{student(x)},$} {\small $\mathit{inPhase(x,prelim)},$} {\small $\mathit{yearsInProgram(x,3)}$},\\
	{\small $\mathit{hardWorking(x)} \leftarrow \mathit{student(x,prelim,3)}.$}\\
	Assume that {\it armg} wants to generalize these clauses to cover example $e'$. Let $e'$ satisfy literal 
	$\mathit{student(x)}$ but does not satisfy $\mathit{inPhase(x,} $ $\mathit{prelim)}$.
	The {\it armg} operator keeps literal 
	$\mathit{student(x)}$ in the first 
	clause, but it eliminates 
	$\mathit{student(x,prelim,3)}$ from the second clause. Hence, it delivers non-equivalent generalizations.
\end{example}
Thus, neither bottom-clause construction nor
generalization phases in ProGolem are schema independent. 
\begin{theorem}
ProGolem is not schema independent.
\end{theorem}


	\section{Castor}
\label{section:castor}
This section presents {\it Castor}, a bottom-up relational learning algorithm.
Castor uses the covering approach presented in Algorithm~\ref{algorithm:generic-covering}.
It follows the same search strategy as ProGolem, 
but integrates INDs into the bottom-clause construction and generalization algorithms to achieve schema independence.
If we apply the INDs in schema $\mR$ to Horn clause $h_{\mR}$ over $\mR$, 
we get an equivalent Horn clause that has a similar syntactic structure 
to its equivalent Horn clauses in decomposition/ compositions of $\mR$ \cite{AliceBook}.
For example, consider schema $\mR_2:$$\{R_1(A,B),R_2(A,C)\}$ with the IND $R_1[A]=$ $R_2[A]$ and the clause 
$h_{\mR_2}:$ $T(x) \leftarrow R_1(x,y)$. 
Because each value in $R_1[A]$ also appears in $R_2[A]$, we can rewrite $h_{\mR_2}$ as $g_{\mR_2}:$ $T(x) \leftarrow R_1(x,y), R_2(x,z)$. 
Now, consider a composition of $\mR$,  $\mS_2:$$\{S_1(A,B,C)\}$.
The clause $h_{\mS_2}:$ $T(x) \leftarrow S_1(x,y,z)$ over $\mS_2$ is equivalent to both $h_{\mR_2}$ and $g_{\mR_2}$. 
$g_{\mR_2}$ and $h_{\mS_2}$ have also similar syntactic structures: 
there is a bijection between the distinct variables in $g_{\mR_2}$ and $h_{\mS_2}$.
However, such bijection does not exist between $h_{\mR_2}$ and $h_{\mS_2}$.
As learning algorithms modify the syntactic structure of clauses to learn a target definition 
and $h_{\mR_2}$ and $h_{\mS_2}$ have different syntactic structures, these algorithms may modify them differently and generate non-equivalent clauses. 
For instance, assume that an algorithm renames variable $z$ 
to $x$ in $h_{\mS_2}$ to generate clause $h^{'}_{\mS_2}:$ $T(x) \leftarrow S_1(x,y,x)$. 
This algorithm cannot apply a similar change to $h_{\mR_2}$ as $h_{\mR_2}$ does not have any corresponding variable to $z$.
But, the algorithm can apply the same modification to $g_{\mR_2}$ and generate an equivalent Horn clause to $h^{'}_{\mS_2}$.
Moreover, as INDs generally reflect important relationships, 
it may improve the effectiveness of the algorithm to use them for learning definitions. 

Castor's $\mathit{LearnClause}$ procedure is shown in Algorithm~\ref{algorithm:castor-beamsearch}. 
It first generates the bottom-clause associated with some positive example using 
the modified bottom-clause construction algorithm presented in Section~\ref{section:CastorBottomClause}.
It minimizes the bottom-clause using the procedure explained in Section~\ref{section:implementation}.
Then, it performs a beam search to select the best 
candidate after multiple applications of the modified {\it ARMG} algorithm, explained in Section~\ref{section:CastorARMG}.
Finally, it reduces the best candidate using the algorithm explained in Section~\ref{section:CastorNegativeReduction}.
\begin{algorithm}
	\SetKwInOut{Input}{Input}
	\SetKwInOut{Output}{Output}
	\Input{Database instance $I$, positive examples $E^+$, negative examples $E^-$, parameters $K$ and $N$.}
	\Output{A new clause $C$. }
	
	$\overrightarrow{C} \leftarrow \mathit{Castor\_BottomClause}(\text{first example in } E^+)$ \;
	
	$\overrightarrow{C} \leftarrow \mathit{Minimize(\overrightarrow{C})}$ ;
	$\mathit{BC} \leftarrow \{ \overrightarrow{C} \}$\;
	
	\Repeat{$\mathit{NC} = \{\}$}{
		$\mathit{BestScore} \leftarrow$ score of highest scoring candidate in $\mathit{BC}$\;
		$E^+_S \leftarrow K$ randomly selected positive examples from $E^+$\;
		$\mathit{NC} = \{\}$\;
		\ForEach{clause $C \in \mathit{BC}$}{
			\ForEach{$e' \in E^+_S$}{
				$C' \leftarrow \mathit{Castor\_ARMG}(C, e')$\;
				\If{$\mathit{Score}(C') > \mathit{BestScore}$}{
					$\mathit{NC} \leftarrow \mathit{NC} \cup C'$\;
				}
			}	
		}
		$\mathit{BC} \leftarrow$ highest scoring $N$ candidates from $\mathit{NC}$\;
	}
	$C' \leftarrow$ highest scoring candidate in $\mathit{BC}$\;

	Return $\mathit{Castor\_Reduce}(C', I, E^-)$\;
	
	\caption{Castor's {\it LearnClause} algorithm. }
	\label{algorithm:castor-beamsearch}
\end{algorithm}
\subsection{Castor Bottom-Clause Construction}
\label{section:CastorBottomClause}
Castor selects a positive example and constructs its bottom-clause by 
following the normal procedure of bottom-clause construction: at each iteration, it selects a relation and 
adds one or more literals of that relation to the bottom-clause.
Let relation symbol $R$ in the schema $\mR$ be decomposed to relation symbols $S_1 \ldots S_n$ in the transformed schema $\mS$. 
If the bottom-clause construction algorithm considers tuple $r$ in an instance of $R$, $I_R$,
it must also examine tuples $s_1, \ldots, s_n$ in instances $I_{S_1}, \ldots, I_{S_n}$, respectively, such that 
$\bowtie_{i=1}^{n}[s_i]$ $=r$, to ensure the produced bottom-clauses over both schemas are equivalent.
After the bottom-clause construction algorithm replaces the constants with variables in these bottom-clauses, 
it generates equivalent bottom-clauses over $\mR$ and $\mS$.
Hence, if Castor examines tuple $s_j \in I_{S_j}$, it should find tuples $s_i \in I_{S_i}$ whose natural join with 
$s_j$ creates tuple $r$. One approach is to find all relations $S_i$ that have some common attributes with $S_j$ 
as they have some tuples that join with $s_i$ and produce $r$. However, designers may rename the attributes on which
$S_1 \ldots S_n$ join. For instance, relations {\it student}, {\it inPhase}, and
{\it yearsInProgram} in the original schema  join over attribute {\it stud} to create relation
{\it student} in the 4NF schema in Table~\ref{table:uwcse}. 
The database designer may rename attribute {\it stud} to {\it name} in relation {\it student}. 
Hence, this approach is not robust against attribute renaming.
According to Definition~\ref{def:decompositionGeneral}, there are INDs with equality between the join attributes of 
relation symbols $S_1 \ldots S_n$.
We use IND with equality between the attributes in schema $\mS$ to find tuples $s_i$.
To simplify our notations, we assume that the join between relations in $\mS$ is still natural join. 
Our results extend for composition joins that are equi-join.
\begin{definition}
The {\it inclusion class} ${\bf N}$ in schema $\mS$ is the maximal set of relation symbols 
in $\mS$ such that for each $S_i, S_j \in {\bf N}$, $i \neq j$, there is 
a sequence of INDs $S_k[X_k]=$ $S'_k[X_k]$, 
$i \leq k \leq j$, in $\mS$ such that
\squishlist
\item $X_k=\mathit{sort}(S_k)$ $\cap \mathit{sort}(S'_k)$.
\item $S_{k+1}=$ $S'_k$ for $i \leq k \leq j-1$.
\end{list}
\end{definition}
\noindent
Castor first constructs the inclusion classes in the input schema $\mS$. 
Assume that the algorithm generates a bottom-clause relative to an instance of schema $\mS$.
Also, assume that the algorithm has just selected relation $I_{S_i}$ and added 
literal $L_i$ to the bottom-clause based on some tuple $s_i$ of $I_{S_i}$.
Let $S_i$ be a member of inclusion class ${\bf N}$ in $\mS$. 
For each constraint $S_j[X]$ $=$ $S_i[X]$ between the members of ${\bf N}$, Castor selects 
all tuples $s_j$ of relation $I_{S_j}$, $i \neq j$ such that $\pi_{X}(s_j)=$ $\pi_{X}(s_i)$.
It applies the same process for $s_j$ until it exhausts the INDs between the members of ${\bf N}$.
As the join between $S_1 \ldots S_n$ is pairwise consistent, this method efficiently 
finds the tuples $s_1, \ldots, s_n$ that all participate in the join and none of them is a dangling
tuple with the regard to the full join.
Otherwise, Castor must check the join condition for each pair of tuples. 
\begin{example}
\label{example:bottom-clause-castor}
Consider an instance of the original UW-CSE schema in Table~\ref{table:uwcse} with tuples
$s_1:\mathit{student(Abe)}$, $s_2:\mathit{inPhase}$\\$\mathit{(Abe,prelim)}$ and  $s_3:\mathit{year(Abe,2)}$.
Given INDs $student[stud]$ $=inPhase[stud]$ and $student[stud]=$ $yearsInProgram[stud]$ hold in this schema, 
{\it student}, {\it inPhase}, and {\it yearsInProgram} constitute an inclusion class.
Let Castor select tuple $s_1$ during the bottom-clause construction.
As $\pi_{stud}(s_1)=$ $\pi_{stud}(s_2)$ and $\pi_{stud}(s_1)$ $=\pi_{stud}(s_3)$, 
Castor adds tuples $s_2$ and $s_3$ to the bottom-clause.
\end{example}
The INDs between relations in a inclusion class may form a cycle.
\begin{definition}
A set of INDs with equality $\lambda$ over schema $\mS$ is {\it cyclic} if
there is a sequence $S_i[X_i]=$ $S'_i[Y_i]$, 
$1 \leq i \leq n$, in $\lambda$ such that
\squishlist
\item $S_{i+1}=$ $S'_i$ for $1 \leq i \leq n-1$ and $S_1=$ $S'_n$.
\item There is an $i$ where $Y_{i}\neq$ $X_{i+1}$.
\end{list}
\end{definition}
\noindent
If the INDs induced by the inclusion class ${\bf N}$ are cyclic, Castor may have to examine a lot more tuples 
than the case where the INDs of ${\bf N}$ are not cyclic.
For example, consider schema $\mS_1$ with relations $S_1(A,B)$, $S_2(B,C)$, and $S_3(C,A)$.
The set of INDs $S_1[B]=$ $S_2[B]$, $S_2[C]=$ $S_3[A]$, and $S_3[A]=$ $S_1[A]$ is cyclic.
Consider tuples $s_1$, $s_2$, and $s_3$ such that $\pi_{B}(s_1)=$ $\pi_{B}(s_2)$ 
and $\pi_{C}(s_2)=$ $\pi_{C}(s_3)$. We may not have
$\pi_{A}(s_3)=$ $\pi_{A}(s_1)$. Hence, Castor has to scan many tuples in 
$S_3$ to find a tuple $s'_3$ that satisfies both $\pi_{C}(s_2)=$ $\pi_{C}(s'_3)$ and 
$\pi_{A}(s'_3)=$ $\pi_{A}(s_1)$.
The following proposition shows that if the composition join in Definition~\ref{def:decompositionGeneral} is acyclic, 
the INDs with equality in the decomposed schema are not cyclic. Thus, Castor does not face the aforementioned issue.
\begin{proposition}
\label{prop:acyclic}
Give schema $\mR$ with a single relation symbol $R$ and its decomposition 
$\mS$ with relation symbols $S_1, \ldots, S_n$,
if the join $\bowtie_{j=1}^{n}[S_1,\ldots,S_n]$ is acyclic, the INDs with equality $\lambda$ in Definition~\ref{def:decompositionGeneral}
are not cyclic.
\end{proposition}
\begin{proof}
	Because the join is acyclic, there is a join tree for it whose nodes are $S_i$, $1 \leq i \leq n$ such that (i) every edge $(S_i,S_j)$ 
	is labeled by the set of attributes $\mathit{sort}(S_i) \cap \mathit{sort}(S_j)$ and (ii)
	for every pair $S_i$, $S_j$ of distinct nodes, for each attribute $A\in$ $\mathit{sort}(S_i) \cap \mathit{sort}(S_j)$,
	each edge along the unique path between $S_i$ and $S_j$ includes label $A$. As the IND with equalities $\lambda$ are defined over
	the common attributes of $S_i$ and $S_j$, $\lambda$ are acyclic.
\end{proof}
\noindent
Given $S_i, S_j \in {\bf N}$, too many tuples from a relation $I_{S_j}$ may join with the current tuple $s_i \in I_{S_i}$, 
which may result in an extremely large bottom-clause. 
One may limit the maximum number of tuples that can join with the current tuple to a reasonably large value.
We use the value of 10 in our reported experiments.
After finding the joint tuples, for each tuple $s_j$, Castor creates a ground literal $L_j$. 
If a constant in $L_j$ has been already seen, the algorithm replaces it in $L_j$ with the variable that was assigned to that constant.
Otherwise, it assigns a fresh new variable for that constant in $L_j$.  
Finally, the algorithm adds $L_j$ to the bottom-clause.
Because inclusion classes are maximal, each relation symbol belongs to at most one inclusion class.
After exhausting all INDs with equality between the members of ${\bf N}$, Castor returns to the typical procedure of bottom-clause construction.
Castor may scan more relations than other bottom-clause construction algorithms to find tuples that satisfy the INDs at the end of each iteration. 
But, a schema usually has a relatively small number of INDs.  We show in Sections~\ref{section:implementation} 
and \ref{section:empirical} that using an RDBMS implementation, Castor bottom-clause construction algorithm runs faster than other algorithms. 

As explained in Section~\ref{section:bottomclause}, the bottom-clauses may get too large. 
We propose a modification of the original bottom-clause construction algorithm so that the stopping condition is based on the maximum number 
of distinct variables in a bottom-clause.
At the end of each iteration, Castor checks how many distinct variables are in the bottom-clause. 
If this number is less than an input parameter, Castor continues to the next iteration and stops otherwise. 
Intuitively, since the number of distinct variables in equivalent Horn clauses over composition/ decomposition are equal, 
this condition helps Castor to return equivalent bottom-clauses over composition/ decomposition.
The following Lemma states that Castor bottom-clause construction algorithm is schema independent. 
\begin{lemma}
	\label{lemma:equivalent-bottom-clauses}
	Let $\tau: \mR \rightarrow \mS$ be a composition/ decomposition, 
	$I$ be an instance
	of $\mR$, and $\perp_{e,I}$ and $\perp_{e,\tau(I)}$ 
	are bottom-clauses generated by Castor for example $e$ relative to $I$ 
	and $\tau(I)$, respectively. We have $\perp_{e,I} \equiv \perp_{e,\tau(I)}$.
\end{lemma}
\begin{proof}
	Assume that 
	$\tau$ decomposes $I_R$ to relations $I_{S_1}, \ldots, I_{S_m}$.
	Let the constants in $e$ appear in 
	a subset of relation $I_R$ denoted as $I^{e}_{R}$.
	Thus, the constants in $e$ must also appear in at least
	a subset of one relation in $\tau(I)_{S_1}, \ldots, \tau(I)_{S_m}$, shown
	as $\tau(I)^{e}_{S_{i}}$, $1 \leq i \leq m$. The algorithm examines all tuples in 
	$I^{e}_{R}$ and $\tau(I)^{e}_{S_{i}}$ at the same iteration.
	Let $L$ be the set of literals that algorithm adds 
	to $\perp_{e,I}$ based on tuples in $I^{e}_{R}$.
	By applying INDs at the end of iteration, 
	the algorithm considers all tuples
	$s_j$ in $I_{S_1}, \ldots, I_{S_m}$ such that
	$\bowtie_{j=1}^{n}[s_i]$ $=r$ for every 
	$r \in I^{e}_{R}$. Hence, it will create equivalent clauses
	at the end of iteration.
	In the following iterations, as the
	algorithm selects tuples in $I$ and $\tau(I)$ 
	using the same set of constants, 
	it adds equivalent literals to the clauses over
	$I$ and $\tau(I)$.
	Because the algorithm uses a one-to-one mapping from variables to 
	constants, the clauses over $I$ and $\tau(I)$
	will be equivalent
	when the algorithm stops. The theorem is similarly proved for 
	composition. Hence, it holds for composition/ decomposition.
\end{proof}

\subsection{Castor Generalization}
\subsubsection{{\it ARMG} Algorithm}
\label{section:CastorARMG}
Castor modifies Algorithm~\ref{algorithm:armg} to compute equivalent {\it armg}s over composition/ decomposition. 
Before we explain the Castor generalization algorithm, we define some concepts.
Given clause $\overrightarrow{C}$ and literal $R(u)$ in $\overrightarrow{C}$,
we call $u$ that may contain both variables and constants a {\it free tuple}. 
We extend the definitions of projection $\pi$ and natural join $\bowtie$ operators over free tuples in natural manner.
A {\it canonical database instance} of clause $\overrightarrow{C}$, shown as $I^{\overrightarrow{C}}$,
is the database instance whose tuples are the free tuples in $\overrightarrow{C}$ \cite{AliceBook}. 
In other words, relation $I_R$ in $I^{\overrightarrow{C}}$ has free tuple $u$ if literal $R(u)$ is in $\overrightarrow{C}$.
In each iteration of the algorithm, Castor ensures that the canonical database instance of clause $\overrightarrow{C}$ always satisfies 
the INDs of the schema.
Assume the algorithm is applied on instance $I_{\mR}$ of schema $\mR=$ $({\bf R}, \Sigma)$.
Immediately after removing a blocking atom $L_i$ from clause $\overrightarrow{C}$ in Algorithm~\ref{algorithm:armg}, 
Castor examines all remaining literals in $\overrightarrow{C}$ and finds the ones whose relation symbols participate 
in an IND with equality in $\Sigma$. More precisely, 
let $R_1(u_1)$ be a literal and $\lambda_{R_1} \subseteq \Sigma$ be the set of INDs with equality in which $R_1$ participates. 
For each IND $R_1[X]=$ $R_2[X]$ in $\lambda_{R_1}$, if there is {\bf not} a literal with relation symbol $R_2$ 
in $\overrightarrow{C}$, Castor eliminates literal $R_1(u_1)$ from $\overrightarrow{C}$.
Otherwise, assume that $\overrightarrow{C}$ contains literal $R_2(u_2)$.
If for all literals $R_2(u_2)$, we have $\pi_{X}(u_1) \neq$ $\pi_{X}(u_2)$, Castor removes literal $R_1(u_1)$. 
Castor checks these conditions for every literal in $\overrightarrow{C}$ and all its corresponding INDs.
Castor increases the time complexity of Algorithm~\ref{algorithm:armg} by a factor of 
$O(|C_{max}|^2  |\lambda|)$, where the $|C_{max}|$ is the size of the largest candidate clause and 
$|\lambda|$ is the number of INDs with equality in the schema.
\begin{example}
	\label{example:armg-castor}
	Consider again the definitions for target relation {\small $\mathit{hardWorking}$} from Example~\ref{example:armg-notrobust} 
 	over the Original and 4NF UW-CSE schemas in Table~\ref{table:uwcse}. 
	Let the INDs $student[stud]=$ $inPhase[stud]$ and 
	$student[stud]=$ $yearsInProgram[stud]$ hold in the Original schema. 
	Assume that Castor wants to generalize these clauses to cover example $e'$, which satisfy 
	$\mathit{student(x)}$ but does not satisfy $\mathit{inPhase(x,prelim)}$.
	Castor removes $\mathit{inPhase}$ literal from the first clause and then removes literals with relation symbols
	$\mathit{student}$ and $\mathit{yearsInProgram}$ due to the INDs 
	in the original schema.
    It also removes $\mathit{student(x,prelim,3)}$ from the second clause.
    Hence, it returns equivalent generalizations.
\end{example}
\begin{lemma}
	\label{armg-schemaindependent}
	Castor's {\it ARMG} is schema independent.
\end{lemma}
\begin{proof}
	Let $\tau: \mR \rightarrow \mS$ be a decomposition from schema $\mR = ({\bf R}, \Sigma_{\mR})$ and $\mS = ({\bf S}, \Sigma_{\mS})$.
	Let $\tau$ map each relation $R^i \in {\bf R}$ to relations $S_{i_1} \ldots $ $S_{i_m} \in {\bf S}$.
	Assume that the input to the {\it ARMG} algorithm over schema $\mR$ is the bottom-clause for seed example $e$, denoted as $\overrightarrow{C_{\mR}}$, which is in the form of $T(w) \leftarrow L_1(u_1), \cdots, L_{n}(u_{n})$. 
	The input to the algorithm over schema $\mS$ is the bottom-clause for seed example $e$, denoted as $\overrightarrow{C_{\mS}}$, which is
	in the following form: $T(w) \leftarrow P_1(v_1), \cdots, P_{k}(v_{k})$. $\overrightarrow{C_{\mR}}$ and $\overrightarrow{C_{\mS}}$ are generated by the Castor bottom-clause construction algorithm and according to Lemma~\ref{lemma:equivalent-bottom-clauses} are equivalent.
	They also do not contain any redundant literal. 
	
	The mapping between equivalent clauses over $\mR$ and $\mS$, $\delta_{\tau}$, that is associated with $\tau$ projects each literal with relation symbol $R^i$ in $\overrightarrow{C_{\mR}}$ to literals with relation symbols $S_{i_1} \ldots S_{i_m}$ in the clause $\overrightarrow{C_{\mS}}$.
	Hence, there is a bijective mapping $M$ that maps each literal $R^i(u_l)$ in the body of $\overrightarrow{C_{\mR}}$ to a set of 
	literals $S_{i_1}(v_j) \ldots$ $S_{i_m}(v_{j + (i_m - i_1)})$ in the body of $\overrightarrow{C_{\mS}}$.
	Moreover, according to Lemma~\ref{lemma:equivalent-bottom-clauses}, 
	a literal $L_l$ appears before $L_o$ in the body of $\overrightarrow{C_{\mR}}$ iff 
	all literals in $M(L_l)$ appear before the ones in $M(L_o)$ in $\overrightarrow{C_{\mS}}$. 	
	The mapping $\delta$ only projects each literal with relation symbol $R^i(u_l)$ to a set of literals in $M(R^i(u_l))$. Hence, 
	the free tuples in every pairs of literals $L_l$ and $L_o$ in $\overrightarrow{C_{\mR}}$ have a variable in common 
	iff the sets of free tuples in $M(L_l)$ and $M(L_o)$ have a shared variable. Otherwise, $\overrightarrow{C_{\mR}}$ and 
	$\overrightarrow{C_{\mS}}$ are not equivalent.

	Assume that Castor removes literal $L_b$ in $\overrightarrow{C}_{\mR}$ because it is the blocking atom in the current iteration.
	Let the positive example considered for this iteration of the algorithm be $e'$. If $L_b$ is the blocking atom, the sub-clause of $\overrightarrow{C}_{\mR}$ up to and excluding $L_b$ covers $e'$ and the one up to and including $L_b$ does not cover $e'$.
	Because mapping $M$ preserves the order of literals, the sub-clause of $\overrightarrow{C}_{\mS}$ up to and excluding $L_b$ 
	covers $e'$ and the one up to and including literals in $M(L_b)$ does not cover it. 
	Hence, at least one literal in $M(L_b)$ is a blocking atom in $\overrightarrow{C}_{\mS}$. 
	If the algorithm removes this literal, it also drops the rest of literals in $M(L_b)$.
	This is because the free tuples of these literals do not satisfy the IND between relation symbols of 
	$M(L_b)$ in the canonical database instance of $\overrightarrow{C}_{\mS}$ after removing the blocking atom in $\overrightarrow{C}_{\mS}$.
	Similarly, if one of the literals in $M(L_b)$ is a blocking atom, $L_b$ will be also a blocking atom. In this case, the {\it ARMG} algorithm will also 
	remove the non-blocking atoms in $M(L_b)$ that are not member of $M(L_o)$, $L_b \neq L_o$ as they do not satisfy any IND 
	after removing the blocking atom.

	Assume that a literal $L_l$ is removed because it does not satisfy any IND in the canonical database 
	instance of $\overrightarrow{C}_{\mR}$ immediately after dropping the blocking atom $L_b$. 
	Let the IND between the relation symbol of $L_b$ and the relation symbol of $L_l$ be $\Sigma_1$.
	Because $\tau$ preserves the INDs between relations in ${\bf R}$, there is also an IND $\Gamma_1$
	between the relation symbol of a literal $P_l$ in $M(L_l)$ and the relation symbol of a literal in $M(L_b)$. 
	Because $L_b$ is a blocking atom, {\it ARMG} algorithm has already removed all literals in $M(L_b)$ from $\overrightarrow{C}_{\mS}$.
	Assume that the free tuples of $P_l$ and another literal $P_o$ in $\overrightarrow{C}_{\mS}$ satisfy $\Gamma_1$.
	If $P_o$ has not been already removed from $\overrightarrow{C}_{\mS}$, 
	the free tuples of $L_l$ and $L_o$ satisfy the IND constraint $\Sigma_1$ in the canonical database of $\overrightarrow{C}_{\mR}$. 
	Thus, $L_l$ should not have been removed from  $\overrightarrow{C}_{\mR}$. 
	Therefore, $P_o$ is removed from $\overrightarrow{C}_{\mS}$.
	Hence, $P_l$ must also be removed from $\overrightarrow{C}_{\mS}$ as it does not satisfy any IND.
	After removing $P_1$, all literals in $M(L_l)$ 
	will be removed from $\overrightarrow{C}_{\mS}$. Using similar argument, we show that if the {\it ARMG} algorithm removes 
	a literal $L_r$ from $\overrightarrow{C}_{\mR}$ because its free tuple does not satisfy any IND after dropping another literal, 
	the algorithm removes the literals in $M(L_r)$ that are not member of $M(L_o)$, $L_r \neq L_o$.
	Also, we prove that if the algorithm eliminates a literal $P_r$ from $\overrightarrow{C}_{\mS}$ 
	because its free tuple does not satisfy any IND, the
	algorithm also removes the literals $L_r$, where $P_r \in M(L_r)$ from $\overrightarrow{C}_{\mR}$.
	We similarly prove that if Castor removes a literal because it is not head-connected, it also removes its corresponding literals over the decomposition and vice versa. 
\end{proof}

\subsubsection{Negative Reduction}
\label{section:CastorNegativeReduction}

\begin{algorithm}
	\SetKwInOut{Input}{Input}
	\SetKwInOut{Output}{Output}
	\Input{Clause $\overrightarrow{C} = T \leftarrow L_1, \cdots, L_n$, database instance $I$, negative examples $E^-$.}
	\Output{Reduced clause $\overrightarrow{C'}$.}
	
	$E^-_c \leftarrow $ subset of $E^-$ covered by $\overrightarrow{C}$\;
	${\bf I} \leftarrow$ list containing all instances of inclusion classes in $\overrightarrow{C}$\;
	\While{true} { 
		$I_i \leftarrow$ first inclusion instance in ${\bf I}$ such that clause $T \leftarrow B$, where $B$ contains literals 
		in inclusion instances $I_1, \cdots, I_i$, has negative coverage $E^-_c$\;
		${\bf H} \leftarrow$ inclusion instances in ${\bf I}$ that connect $I_i$ with $T$\;
		${\bf N} \leftarrow$ literals from inclusion instances $I_1, \cdots, I_i$ not in ${\bf H}$\;
		${\bf I'} \leftarrow {\bf H} \cup [ I_i ] \cup {\bf N}$ \;
		\If{$length({\bf I'}) = length({\bf I})$} {
			$C' = T \leftarrow B$, where $B$ contains all literals in ${\bf I'}$\;
			Return $C'$\;
		}
		${\bf I} \leftarrow {\bf I'}$\;
	}
	\caption{Castor negative reduction algorithm.}
	\label{algorithm:reduction}
\end{algorithm}

Castor further generalizes clauses produced by {\it ARMG} by removing non-essential literals from clauses.
A literal is {\it non-essential} if after it is removed from a clause, 
the number of negative examples covered by the clause does not increase~\cite{golem,progolem}. 
This step is called {\it negative reduction} and reduces the generalization error of the produced definitions to the training data. 
Castor uses INDs with equality to compute equivalent reductions of clauses over composition/ decomposition.
Given a clause $\overrightarrow{C}$ and inclusion class ${\bf N}=$ $\{S_i \mid 1 \leq i \leq m\}$ 
over schema $\mS$, an instance $Y_{\bf N}$ of ${\bf N}$ is a set of literals $S_1(u_1), \cdots, S_m(u_m)$ 
in $\overrightarrow{C}$ such that for every IND $S_i[X] = S_j[X]$, $1 \leq i,j \leq m$, 
there are literals $S_i(u_i)$ and $S_j(u_j)$ in $Y_{\bf N}$ such that $\pi_{X}(u_i) =$ $\pi_{X}(u_j)$.
An instance $Y_{\bf N}$ over a clause $\overrightarrow{C}$
is {\it non-essential} if after removing all literals in $Y_{\bf N}$ from $\overrightarrow{C}$, 
the number of negative examples covered by the clause does not increase.
First, for each literal $L_j$ in the input clause $\overrightarrow{C}$, Castor computes the instances of inclusion classes in $\overrightarrow{C}$ that start with $L_j$.
It creates a list containing all found instances, in the order in which they are found.  
Then, it iteratively removes non-essential instances from this list.
In each iteration, it finds the first inclusion instance $Y_i$ 
such that the sub-clause of $\overrightarrow{C}$ that contains all literals in every inclusion instance up to $Y_i$ 
has the same negative coverage as $\overrightarrow{C}$.
A head-connecting inclusion instance for $Y_i$ contain literals that connect a literal in $Y_i$ to the head of the clause by a chain of variables.
Castor moves $Y_i$ and its head-connecting inclusion instances to the beginning of the list, and discards the inclusion instances after $Y_i$.
These instances can be discarded because they are non-essential.
Note that some literals in the discarded instances may also belong to other instances before or in $Y_i$.
The algorithm iterates until the number of inclusion instances in the clause does not change after one iteration. 
At the end, it creates a clause whose head literal is the same as $\overrightarrow{C}$ and 
body contains all literals in the remaining instances of inclusion classes.
Because negative reduction only removes literals from the clause, it does not decrease the number of positive examples covered by 
the clause.
More details can be found in Algorithm~\ref{algorithm:reduction}. 
\begin{lemma}
	\label{castor-reduction-schemaindependent}
	Castor's negative reduction is schema independent.
\end{lemma}
\begin{proof}
	Let $\tau: \mR \rightarrow \mS$ be a composition/ decomposition
	between schemas
	$\mR = ({\bf R}, \Sigma_{\mR})$ and $\mS = ({\bf S}, \Sigma_{\mS})$.
	Let $R[U] \in {\bf R}$ and
	$\tau$ map relation $R[U]$ to
	relations $S_1[V_1], \cdots, $ $ S_m[V_m]$, 
	$1 \le m \le |{\bf S}|$.
	Let ${\bf N}$ be the inclusion class in $\Sigma_{\mS}$ that contains relations $S_1[V_1], \cdots, $ $ S_m[V_m]$.	
	Assume that $\overrightarrow{C}_{\mR}$ is a clause over
	schema $\mR$ and contains $k$ literals $R(u_i)$, 
	$1 \leq i \leq k$.
	Let $\overrightarrow{C}_{\mS}$ be the equivalent clause of
	$\overrightarrow{C}_{\mR}$ over $\mS$.
	Let $\mathit{Reduce}(C)$  be the function that performs negative reduction on clause $C$.
	We show that 
	$\mathit{Reduce}(\overrightarrow{C}_{\mR}) \equiv \mathit{Reduce}(\overrightarrow{C}_{\mS})$.
	
	Because $\overrightarrow{C}_{\mR}$ contains $k$ literals $R(u_i)$, $1 \leq i \leq k$, and $\overrightarrow{C}_{\mR} \equiv \overrightarrow{C}_{\mS}$, 
	then $\overrightarrow{C}_{\mS}$ must contain $k$ instances of inclusion class ${\bf N}$.
	These instances of inclusion class may or may not share literals.
	Let $n$ be the number of instances of inclusion class ${\bf N}$ in $\overrightarrow{C}_{\mS}$ that share literals. 
	Without loss of generality, we assume that instances can only share the first literal. 
	That is, instances $I_{{\bf N}i}$ and $I_{{\bf N}j}$ share a literal if they have the form 
	$I_{{\bf N}i} = $ $S_1(v_{\_1}), S_2(v_{i2}), \cdots, $ $S_m(v_{im})$ and 
	$I_{{\bf N}j} = $ $S_1(v_{\_1}), S_2(v_{j2}), \cdots, $ $S_m(v_{jm})$.
	We prove by induction on $n$. 
	
	Base case: let $n=1$. Clause $\overrightarrow{C}_{\mR}$ contains literal $R(u)$ and 
	$\overrightarrow{C}_{\mS}$ contains an instance of inclusion class ${\bf N}$ with literals $S_1(v_{1}), \cdots, $ $ S_m(v_{m})$ such that 
	$\bowtie_{l=1}^{m}[v_{l}]$ $=u$.
	Notice that $\overrightarrow{C}_{\mR}$ may contain other literals with relation $R$ and $\overrightarrow{C}_{\mS}$ may contain other instances of inclusion class ${\bf N}$.
	However, because $n=1$, these instances do not share literals and can be treated independently.
	Then, Castor removes literal $R(u)$ in 
	$\overrightarrow{C}_{\mR}$ iff it removes literals
	$S_1(v_{1}), \cdots, $ $ S_m(v_{m})$ in $\overrightarrow{C}_{\mS}$.
	
	Assumption step: let $n=k$. 
	$\overrightarrow{C}_{\mR}$ contains literals $[R(u_{i})]$, $1 \le i \le k$,
	$\overrightarrow{C}_{\mS}$ contains literals $S_1(v_{\_1}), [S_2(v_{i2}), \cdots, $ $S_m(v_{im})]$, $1 \le i \le k$ and
	$\overrightarrow{C}_{\mR} \equiv \overrightarrow{C}_{\mS}$.
	
	Induction step: let $n=k+1$.
	Let $\overrightarrow{C}_{\mS}$ contain $k+1$ instances of inclusion class ${\bf N}$, which share the first literal.
	Let $\overrightarrow{C}_{\mR}$ be the equivalent clause, which contains $k+1$ literals $R(u_{i})$, $1 \le i \le k+1$.
	We divide instances in $\overrightarrow{C}_{\mS}$ in two: ${\bf I_{{\bf N}(1..k)}} = $ $S_1(v_{\_1}), [S_2(v_{i2}), \cdots, $ $S_m(v_{im})]$, $1 \le i \le k$ and $I_{{\bf N}(k+1)} =$ $S_1(v_{\_1}), S_2(v_{(k+1)2}), \cdots, $ \\
	$S_m(v_{(k+1)m})$.
	We also divide literals in $\overrightarrow{C}_{\mR}$ in two: ${\bf R_{1..k}} = [R(u_{i})]$, $1 \le i \le k$ and $R(u_{k+1})$.
	
	Let $\overrightarrow{C'}_{\mS}$ contain all literals in ${\bf I_{{\bf N}(1..k)}}$ and $\overrightarrow{C'}_{\mR}$ contain all literals in ${\bf R_{1..k}}$.
	We examine the cases where we add literal $R(u_{k+1})$ to $\overrightarrow{C'}_{\mR}$ such that $\overrightarrow{C'}_{\mR} \cup \{ R(u_{k+1}) \} = \overrightarrow{C}_{\mR}$, 
	and we add all literals in instance $I_{{\bf N}(k+1)}$ to $\overrightarrow{C'}_{\mS}$ such that $\overrightarrow{C'}_{\mS} \cup I_{{\bf N}(k+1)} = \overrightarrow{C}_{\mS}$.
	
	Castor removes all literals in ${\bf R_{1..k}}$ and literal $R(u_{k+1})$ iff 
	it removes all literals in ${\bf I_{{\bf N}(1..k)}}$ and $I_{{\bf N}(k+1)}$.
	Then, $\mathit{Reduce}(\overrightarrow{C}_{\mR}) \equiv$ $\mathit{Reduce}(\overrightarrow{C}_{\mS})$.
	
	Castor removes all literals in ${\bf R_{1..k}}$ but not literal $R(u_{k+1})$ iff 
	it removes all literals in ${\bf I_{{\bf N}(1..k)}}$, but not literals in $I_{{\bf N}(k+1)}$.
	Notice that literal $S_1(v_{\_1})$ stays in clause $\mathit{Reduce}(\overrightarrow{C}_{\mS})$ because it is in instance $I_{{\bf N}(k+1)}$.
	Because $\tau (R(u_{k+1})) = $ $S_1(v_{\_1}), S_2(v_{(k+1)2}), \cdots, $ $S_m(v_{(k+1)m})$, 
	then $\mathit{Reduce}(\overrightarrow{C}_{\mR}) \equiv$ $\mathit{Reduce}(\overrightarrow{C}_{\mS})$.
	
	Castor removes literal $R(u_{k+1})$ but not literals in ${\bf R_{1..k}}$ iff 
	it removes all literals in $I_{{\bf N}(k+1)}$, but not literals in ${\bf I_{{\bf N}(1..k)}}$.
	Again, notice that literal $S_1(v_{\_1})$ stays in clause $\mathit{Reduce}(\overrightarrow{C}_{\mS})$ because it is in instances ${\bf I_{{\bf N}(1..k)}}$.
	Because we know that $\mathit{Reduce}(\overrightarrow{C'}_{\mR}) \equiv$ $\mathit{Reduce}(\overrightarrow{C'}_{\mS})$ (assumption step), 
	then $\mathit{Reduce}(\overrightarrow{C}_{\mR}) \equiv$ $\mathit{Reduce}(\overrightarrow{C}_{\mS})$.
\end{proof}
\noindent
Based on 
Lemmas~\ref{lemma:equivalent-bottom-clauses},
\ref{armg-schemaindependent}, and
\ref{castor-reduction-schemaindependent}, 
Castor is schema independent.

\subsection{Generating Safe Clauses}
\label{section:castor-safe-clauses}
Let the head-variables of a clause be the ones that appear in its head literal.
A clause is {\it safe} if every head-variable appears in some literal in the body of the clause.
A safe definition is safe if all its clauses are safe.
The results of safe clauses and definitions are finite over a (finite) database.
By default, current relational learning algorithms, including Castor, may learn {\it unsafe} Datalog definitions \cite{AliceBook}.
Because an unsafe definition produces infinitely many answers over a (finite) database, it is {\it not} desirable in many relevant applications, such as learning database queries from examples~\cite{Maier:VLDB:2015,Abouzied:PODS:13}.
Furthermore, a relational learning algorithm that learns only safe clauses can learn a definition from positive examples only.
In this section, we describe how Castor can be modified to generate only safe definitions.
As we have explained, Castor first constructs the bottom-clause associated with some positive example $e$, and 
then generalizes this clause using {\it ARMG} and negative reduction.

\subsubsection{Bottom-clause Construction:} 
The bottom-clause construction uses the positive example $e$ as the initial head-literal for the bottom-clause. Castor picks every literal in body of the bottom-clause based on the constants/ variables in the head-literal. Thus, the bottom-clause is guaranteed to be safe.

\subsubsection{Safe {\it ARMG} Algorithm}
Castor's {\it LearnClause} procedure calls the {\it ARMG} algorithm multiple times. 
Let the {\it ARMG} algorithm take as input clause $\overrightarrow{C}$ and positive example $e$, and produce as output clause $\overrightarrow{C'}$.
Clause $\overrightarrow{C'}$ may not be safe.
Besides checking whether the score of $\overrightarrow{C'}$ is better than the current best score, Castor also checks whether $\overrightarrow{C'}$ is safe.
If $\overrightarrow{C'}$ is safe, Castor considers it as a candidate.
If $\overrightarrow{C'}$ is not safe, Castor simply ignores it.

\subsubsection{Safe Negative Reduction}
In negative reduction, Castor first computes all instances of inclusion classes, and then iteratively removes non-essential instances.
In order to output a safe clause, Castor first sorts all instances of inclusion classes by the number of head-variables appearing in the instance in descending order.
Then, in each iteration, Castor finds the first inclusion instance $Y_i$ 
such that the sub-clause of $\overrightarrow{C}$ that contains all literals in every inclusion instance up to $Y_i$ 
has the same negative coverage as $\overrightarrow{C}$.
Castor then finds the head-connecting inclusion instances for $Y_i$.
Let these instances be called ${\bf H_{Y_i}}$.
Next, from the instances of inclusion classes that will be discarded, Castor finds the first instances that contain head-variables that do not appear in $Y_i$ or ${\bf H_{Y_i}}$.
Let these instances be ${\bf S_{Y_i}}$
The goal is to find literals needed to make the resulting clause safe.
These literals are guaranteed to exist because the clauses produced by {\it ARMG} are forced to be safe.
Castor then moves $Y_i$, ${\bf H_{Y_i}}$, and ${\bf S_{Y_i}}$ to the beginning of the list, and discards the inclusion instances after $Y_i$, except the ones in ${\bf S_{Y_i}}$.
The algorithm continues its normal operation until the number of inclusion instances in the clause does not change.
Finally, it creates a clause whose body contains all literals in the remaining instances of inclusion classes.

\subsection{General Decomposition/ Composition}
\label{sec:generalDecomposition}
Castor is robust over schema variations caused by bijective decompositions and compositions as defined in Section~\ref{transformations}.
Bijective decompositions and compositions need at least one IND with equality in the transformed and original schemas, respectively.
We have observed several examples of these transformations in real-world databases, some of which we report in Section~\ref{section:empirical}. 
However, in addition to INDs with equality, schemas often have INDs in the general form of subset or equality.
One can use these INDs to define a more general decomposition.
More precisely, a {\it general decomposition} of schema $\mR$ with single relation symbol $R$
is schema $\mS$ with relation symbols $S_1\ldots$ $S_n$ that satisfies all conditions in Definition~\ref{def:decompositionGeneral} 
but at least one IND in $\mS$ (in the second condition of Definition~\ref{def:decompositionGeneral}) is an IND in form of subset or equality.
A general decomposition of a schema with multiple relations is the union of general decompositions over each relation symbol in the schema.

A general decomposition is invertible but not bijective \cite{AliceBook}.
Consider the general decomposition from $\mR_1:$$\{R_1(A,B,C)\}$ to \\
$\mS_1:$$\{S_1(A,B), S_2(A,C)\}$ with IND $S_2[A]\subseteq$ $S_1[A]$, 
and the instance of $\mS_1$ $I^{1}_{\mS_1}:$ $I^{1}_{S_1}=\{(a_1,b_1),(a_2,b_2)\}$, $I^{1}_{S_2}=\{(a_1,c_1)\}$.
There is {\it not} any instance of $\mR_1$ that represents the same information as $I^{1}_{\mS_1}$. 
Hence, it is not clear how to define schema independence for $I^{1}_{\mS_1}$.
Also, the composition from $\mS_1$ to $\mR_1$ is not invertible as $I^{1}_{S_1} \bowtie I^{1}_{S_2}$ loses tuple $(a_2,b_2)$, 
which cannot be recovered. 
As some original and transformed databases in this composition do not have the same information,
it is not reasonable to expect equivalent learned definitions over these databases.

One may resolve these issues by considering databases with labeled nulls, e.g., 
by using weak universal relation assumption \cite{AliceBook,Fagin-inverse}.
For example, one can compose instance $I^{1}_{\mS_1}$ in the last example to $I^{1}_{\mR_1}:$ $\{(a_1, b1, c_1),$ $(a_2, b2, x)\}$ where
$x$ is a labeled null that reflects the existence of an unknown value. 
However, it takes more than a single paper to define the semantic of learning 
over databases with labeled nulls and schema independence over transformations that introduce labeled nulls, 
so we leave this direction for future work. 
Instead, we define schema independence for general decompositions by ignoring the instances in the transformed schema that do not have any
corresponding instance in the original schema.
Hence, the mapping between the instances in the original and the remaining instances of the transformed schemas is bijective, thus, it is 
definition bijective.
We define hypothesis invariance and schema independence as defined in Section~\ref{section:framework} for this mapping. 
An algorithm is schema independent over a general decomposition if it is schema independent over its mapping between 
the corresponding instances of the original and decomposed schemas.

A {\it general composition} is the inverse of a general decomposition.
As we have shown, general compositions lose information. Thus, it is not reasonable to expect algorithms to be schema independent over them.
We limit the instances of its original schema so that it becomes invertible.
For simplicity, we define schema independence for a general composition whose transformed schema has a single relation. 
Our definition extends for schemas with multiple relations.
Let schema $\mR$ with a single relation symbol $R$ be a general composition of schema $\mS$ with relation symbols 
$S_1\ldots$ $S_n$ such that for all $S_i, S_j$, $1 \leq i,j \leq n$, $X =$ $sort(S_i) \cap sort(S_j)$ $\neq \emptyset$,
$\mS$ has IND $S_i[X] \subseteq S_j[X]$.
Natural join between $S_1\ldots$ $S_n$ does not lose any tuple in an instance of $\mS$, $I_{\mS}$, iff
for each IND $S_i[X] \subseteq S_j[X]$ in $\mS$ we have $\pi_{X}(I_{S_i})=$ $\pi_{X}(I_{S_j})$, where
$I_{S_i}$ and $I_{S_j}$ are relations of $S_i$ and $S_j$ in $I_{\mS}$, respectively.
Let $J(\mS)$ denote instances with the aforementioned property in $\mS$. 
The mapping from $J(\mS)$ to $I(\mR)$ is bijective, therefore, 
it is definition bijective. Thus, hypothesis invariance and schema independence properties in Section~\ref{section:framework}
can be define for this mapping.
An algorithm over the general composition from $\mS$ to $\mR$ is schema independent 
if it is schema independent over the mapping between $J(\mS)$ to $I(\mR)$.
We call a finite application of general decompositions and compositions a general decomposition/ composition.
An algorithm is schema independent over a general decomposition/ composition if it is schema independent over
its general decompositions and general compositions.

Consider again schema $\mS$ with relation symbols 
$S_1\ldots$ $S_n$.
To achieve schema independence over general composition/ decomposition, given instance $I_{\mS}$, 
Castor finds each INDs $S_i[X] \subseteq S_j[X]$ in $\mS$ where $\pi_{X}(I_{S_i})=$ $\pi_{X}(I_{S_j})$ and adds the IND to its list of IND with equality in a preprocessing step. 
It then proceeds to its normal execution. The proofs of lemmas~\ref{lemma:equivalent-bottom-clauses}, \ref{armg-schemaindependent}, and \ref{castor-reduction-schemaindependent} extend for the corresponding instances of $\mR$ and $\mS$ that have the same information in non-bijective decompositions. Using a similar argument, these proofs also hold for the corresponding instances that have the same information over general decomposition. 
Thus, Castor is schema independent over general decompositions/ compositions. 
Using this method, Castor also handles combinations of INDs in general form and INDs with equality.

The pre-processing step of checking for each IND  $S_i[X] \subseteq S_j[X]$
in schema $\mS$ whether $\pi_{X}(I_{S_i})=$ $\pi_{X}(I_{S_j})$ holds, may 
take a long time and some users may not want to wait for this pre-processing phase to finish. 
Another approach is to use INDs in form of subset or equality in Castor directly as follows.
We extend Castor to use both INDs with equality and in general form. 
In the rest of this section, we refer to both type of INDs simply as IND and write them by $\subseteq$ for brevity.
We redefine an inclusion class ${\bf N}$ in schema $\mS$ as a set of relation symbols $S_i, S_j$ in $\mS$
such that there is a sequence of INDs $S_k[X_k]\subseteq$ $S'_k[X_k]$ or $S'_k[X_k]\subseteq$ $S_k[X_k]$
$i \leq k \leq j$, in $\mS$ where $X_k=\mathit{sort}(S_k)$ $\cap \mathit{sort}(S'_k)$ and $S_{k+1}=$ $S'_k$ for $i \leq k \leq j-1$.
Assume that Castor picks a tuple $s_i$ from relation $S_i$ in inclusion class ${\bf N}$ during the bottom-clause construction.
For each $S_i[X] \subseteq$ $S_j[X]$ in ${\bf N}$, 
Castor selects all tuples $s_j$ of relation $I_{S_j}$, $i \neq j$ such that $\pi_{X}(s_j)\subseteq$ $\pi_{X}(s_i)$.
Castor repeats this process for $s_j$ until it exhausts all INDs in ${\bf N}$. 
After this step, Castor follows the bottom-clause construction algorithm explained in Section~\ref{section:CastorBottomClause}.
Since the natural join between relations in $\mS$ is acyclic, the pairwise consistency implies the global consistency of the joint tuples. 
For the same reason, the proof of Proposition~\ref{prop:acyclic} extends for INDs. 
Hence, the INDs in each inclusion class are not cyclic and Castor efficiently finds the tuples that join according to the INDs. 
We also extend the Castor ARMG algorithm to ensure that the free tuple of each literal $S(u)$, $u$,
satisfies all INDs in which $S$ participates after a blocking atom is removed. 
If $u$ does not satisfy any of its corresponding INDs, it is removed. 
Finally, we redefine the instance of an inclusion class ${\bf N}$, $Y_{\bf N}$, in an ordered clause $\overrightarrow{C}$ as
a set of literals $S_1(u_1), \cdots, S_m(u_m)$ in $\overrightarrow{C}$ such that for each IND $S_i[X] \subseteq S_j[X]$, 
$1 \leq i,j \leq m$, there are literals $S_i(u_i)$ and $S_j(u_j)$ in $Y_{\bf N}$ where $\pi_{X}(u_i) =$ $\pi_{X}(u_j)$. 
We modify our negative reduction algorithm in Section~\ref{section:CastorNegativeReduction} to use the new definition of inclusion class instance.
This extension of Castor may not be schema independent as it may miss some tuples in bottom-up construction or ignore some literals in ARMG algorithms. 
For example, consider the general decomposition from $\mR_1:$$\{R_1(A,B,C)\}$ to  
$\mS_1:$$\{S_1(A,B), S_2(A,C)\}$ with IND $S_2[A]\subseteq$ $S_1[A]$ and instances
$J^{1}_{\mR_1}:$ $J^{1}_{R_1}=\{(a_1,b_1,c_1)\}$ and $J^{1}_{\mS_1}:$ $J^{1}_{S_1}=\{(a_1,b_1)\}$, $J^{1}_{S_2}=\{(a_1,c_1)\}$. 
Assume that the modified Castor bottom-clause construction over $J^{1}_{\mS_1}$ starts with tuple 
$(a_1,b_1)$. IND $S_2[A]\subseteq$ $S_1[A]$ does not force Castor to select $(a_1,c_1)$ for the bottom-clause.
Hence, Castor delivers non-equivalent bottom-clauses over $J^{1}_{\mS_1}$ and $J^{1}_{\mR_1}$.
But, our empirical results in Section~\ref{section:empirical} show that this extension of Castor 
is more schema independent than other algorithms over general decomposition/ composition.


\subsection{Castor System Design Choices and Implementation}
\label{section:implementation}
Current bottom-up algorithms do not run efficiently over medium or large databases because they
produce many long bottom-clauses to generalize \cite{progolem}.
Also, these clauses are time-consuming to evaluate.
A relational learning algorithm evaluates a clause by computing the number of positive and negative examples covered by the clause.
These tests dominate the time for learning~\cite{DeRaedt:LogicalLearning:08}.
It is generally time-consuming to evaluate clauses with many literals. 
Castor implements several optimizations to run efficiently over large databases. 

\subsubsection{In-memory RDBMS}
Castor is implemented on top of the in-memory RDBMS VoltDB ({\it voltdb.com}).
Relational databases are usually stored in RDBMS's. Therefore, it is natural to implement a learning algorithm on top of an RDBMS.
Using an RDBMS also provides access to the schema constraints, e.g., inclusion dependencies, which we use to achieve schema independence. 
Castor performs bottom-clause construction multiple times during the learning process.
The bottom-clause construction algorithm queries the database multiple times each of which selects all tuples in a table that match given constants from the training data. We leverage RDBMS indexing to
improve the running time of these queries.

\subsubsection{Stored Procedures}
We implement the bottom-clause construction algorithm inside a stored procedure to reduce the number of API calls made from Castor to the RDBMS. Castor makes only one API call per each bottom-clause.
The first time that Castor is run on a schema, it creates 
the stored procedure that implements the bottom-clause construction algorithm for the given schema.
Castor reuses the stored procedure when the algorithm is 
run again, with either new training data or updated database instance. 

\subsubsection{Efficient Clause Evaluation}
\label{section:clause-evaluation}
One approach to computing the number of positive (negative) examples covered by a clause is to join the table containing the positive (negative) examples with the tables corresponding to all literals in the body of the clause.
If two literals share a variable, then a natural join between the two columns corresponding to the shared variable in the literals is used.
This strategy works well when clauses are short, as in top-down algorithms~\cite{QuickFOIL}.
However, our empirical studies show that the time and space requirements for this approach are prohibitively large on large clauses generated by bottom-up algorithms.
Thus, we perform coverage tests by using a subsumption engine.
Clause~$C$ {\it $\theta$-subsumes} $C'$ iff there is some substitution $\theta$ such that $C \theta \subseteq C'$.
A ground bottom-clause is a bottom-clause that only contains constants. 
A candidate clause~$C$ covers example $e$ iff $C$ $\theta$-subsumes the ground bottom-clause $\perp_{e}$ associated with $e$.
Castor uses the efficient subsumption engine Resumer2~\cite{Kuzelka:2009:RSE:1497096.1497102}. 
Resumer2 efficiently checks if clause $C$ covers example $e$ by deciding the subsumption between $C$ and the ground bottom-clause $\perp_{e}$ of $e$.
Given clause $C$ and a set of examples $E$, Castor checks if $C$ covers each $e \in E$ separately. 
Castor divides $E$ in subsets and performs coverage testing for each subset in parallel.

\subsubsection{Coverage Tests}
Castor optimizes the generalization process by reducing the number of coverage tests.
Castor first generates the bottom-clause relative to a positive example. Then, Castor generalizes this clause in the beam search and ARMG algorithms by generating new, more general clauses.
If clause $C$ covers example $e$, then clause $C''$, which is more general than
$C$, also covers $e$. If Castor knows that $C$ covers $e$, it does not check if $C''$ covers $e$.

\subsubsection{Minimizing Clauses}
Bottom-up algorithms such as Castor produce large clauses, which are expensive to evaluate.
Castor minimizes bottom-clauses by removing syntactically redundant literals.
A literal $L$ in clause $C$ is {\it redundant} if $C$ is equivalent to $C' = C - \{L\}$.
Clause equivalence between $C$ and $C'$ can be determined by checking whether $C$ $\theta$-subsumes $C'$ and $C'$ $\theta$-subsumes $C$. 
Castor minimizes clauses using theta-transformation~\cite{Costa:2003:QTI:945365.945392}.
It uses a polynomial-time approximation of the clausal-subsumption test, which is efficient and retains the property of correctness.
Given clause $C$, for each literal $L$ in $C$, the algorithm checks if $C \subseteq C' = C - \{L\}$. 
If this holds, then $L$ is redundant and will be removed.
Minimizing bottom-clauses reduces the hypothesis space considered by Castor.
It also makes coverage testing faster.
Castor also minimizes learned clauses before adding them to the definition. This ensures that clauses are concise and interpretable.

	\section{Query-based algorithms}
\label{section:worstcase}

In this section, we consider query-based learning algorithms, which learn exact definitions by asking queries to an oracle~\cite{Khardon:1999,Reddy98learninghorn,Selman:2011,Arias:2007,Abouzied:PODS:13}.
This type of algorithms have been recently used in various areas of database management, such as finding schema mappings and designing usable query interfaces  \cite{Cate:2013:LSM:2539032.2539035,Abouzied:PODS:13}.
Queries can be of multiple types, however the most common types are equivalence queries and membership queries.
In equivalence queries (EQ), the learner presents a hypothesis to the oracle and the oracle returns {\it yes} if the hypothesis is equal to the target relation definition, otherwise it returns a counter-example. 
In membership queries (MQ), the learner asks if an example is a positive example, and the oracle answers {\it yes} or {\it no}.

Because query-based algorithms follow a different learning model, Definition \ref{schema_independence} is not suited for evaluating their schema (in)dependence. 
Since a query-based algorithm can ask the oracle whether candidate definitions are correct,  
the algorithm will always learn the correct definitions by asking sufficient number of queries from the oracle.
As it takes time and/or resources to answer queries, a desirable query-based algorithm should not ask too many queries \cite{Arias:2007}.
For instance, some database query interfaces use query-based algorithms to
discover users' intents \cite{Abouzied:PODS:13}. 
Because the oracle for these algorithms is the user of the database, a more desired algorithm 
should figure out the user's intent by asking fewer queries from the user.

Query-based algorithms are theoretically evaluated by their {\it query complexity} 
-- the asymptotic number of queries asked by the algorithm \cite{Khardon:1999}.
Therefore, we analyze the impact of schema transformations on the query complexity of these algorithms.
Generally, if an algorithm has different asymptotic behavior over equivalent schemas, then the algorithm is schema dependent.
One way to show that an algorithm has different asymptotic behavior over different schemas is by comparing the lower bound on the query complexity of the algorithm against the upper bound on its query complexity.
If the lower bound under one of the schemas is greater than the 
upper bound under another schema, then the algorithm is highly schema dependent. 
Of course, this is not a desirable property, as this means that the choice of 
representation has a huge impact on the performance of the algorithm. 
However, we prove that a popular query-based algorithm called $A2$ suffers from this property.

$A2$~\cite{Khardon:1999} is a query-based learning algorithm that learns function-free, first-order Horn expressions.
The reasons for choosing this algorithm are three fold: i) $A2$ is representative of query-based learning algorithms that work on the relational model, ii) there is an implementation of the algorithm~\cite{Arias:2007}, iii) $A2$ is a generalization to the relational model of a classic query-based propositional algorithm~\cite{Angluin:1992}.

\begin{theorem}
	\label{A2_schema_independent}
	Let $\Omega (f)_{\mR}$ and $O (g)_{\mR}$ be the lower bound and upper bound, respectively, 
	on the query complexity of $A2$ for all target relations under schema $\mR$,
	where $f$ and $g$ are functions of properties of $\mR$.
	Then, there is a composition/ decomposition of $\mR$, $\mS$, such that $\Omega (f)_{\mR} > O(g)_{\mS}$.
\end{theorem}
\begin{proof}
	Let $\mR$ and $\mS$ be two definition equivalent schemas.
	Schema $\mR = ({\bf R}, \Sigma)$ contains the single relation $R(A_1, \cdots, A_l)$.
	Assume that $l \ge 2$ and there are
	$l-1$ functional dependencies $A_1 \rightarrow A_i$, $2 \leq i \leq l$, in $\Sigma$.
	Let $\mS = ({\bf S}, \Omega)$ be a vertical decomposition of schema $\mR$, such that
	relation $R(A_1, \cdots, A_l) \in$ ${\bf R}$ is decomposed into
	$l-1$ relations in ${\bf S}$ in the form of $S_i(A_1,A_i)$, $2 \leq i \leq l$.
	For each relation $S_i(A_1,A_i) \in {\bf S}$, $\Omega$ contains the functional dependency
	$A_1 \rightarrow A_i$. For each set of relations $S_i(A_1,A_i)$, $2 \leq i \leq l$,
	$\Omega$ also contains $2(l-1)$ inclusion dependencies in the form of $S_2.A_1 \subseteq S_j.A_1$
	and $S_j.A_1 \subseteq S_2.A_1$, $2 < j \leq l$.
	
	Let $p_{\mR}$ be the number of relations in schema $\mR$, $a_{\mR}$
	be the largest arity of any relation in $\mR$, $k_{\mR}$ be the largest number
	of variables in a clause, and $m_{\mR}$ be the number of clauses in the definition of the target relation over $\mR$.
	We define $p_{\mS}$, $a_{\mS}$, $k_{\mS}$, and $m_{\mS}$ analogously.
	The largest number of constants (i.e. objects) in any example is denoted by $n$. Parameter $n$ is a constraint on the answers of the oracle, therefore it is independent of the hypothesis space and the schemas.
	Because the number of relations in $\mR$ is $p_{\mR}=1$ and the maximum arity is $a=a_{\mR}$,
	then the maximum number of relations in $\mS$ is $p_{\mS} = a-1$. We also have that $a_{\mS} = 2$.
	
	
	Let $\mL$ be the hypothesis language that consists of the subset of Horn definitions that contain a single clause in which no self-joins are allowed.
	All definitions in $\mL$ under schema $\mR$ have the form
	\begin{equation*}
	T({\bf u}) \leftarrow R(x_1, x_2, \cdots, x_l).
	\end{equation*}
	where $T$ is the target relation and ${\bf u}$ is a subset of $\{ x_1, x_2, \cdots, x_l \}$.
	
	Any clause in a definition $h_{\mR} \in \mL$ under schema $\mR$ has at most $l$ distinct variables, which corresponds to the arity of relation $R$. Therefore $k_{\mR} = l$.
	As schema $\mS$ is a vertical decomposition of schema $\mR$, and no self-joins are allowed in $\mL$,
	the definition $\delta (h_{\mR}) = h_{\mS} \in \mL$ also has at most $l$ variables.
	We will use $k = k_{\mR}$ to denote the upper bound on $k_{\mR}$ and $k_{\mS}$.
	Because definitions in $\mL$ consist of a single clause, then the maximum number of clauses in a definition $m = 1$. In general, $m_{\mR} = m_{\mS}$ because $\mS$ is a vertical decomposition of $\mR$.
	
	The upper bound on the number of EQs and MQs in the $A2$ algorithm is $O(m^2pk^{a+3k} + nmpk^{a+k})$, and the lower bound is $\Omega (mpk^a)$ \cite{Khardon:1999}.
	In order to prove our theorem, the following should hold for $\mR$ and $\mS$
	\begin{align*}
	\Omega(mp(k_{\mR})^a)_{\mR} > O(m^2p(a-1)(k_{\mS})^{2+3k_{\mS}} + nmp(a-1)(k_{\mS})^{2+k_{\mS}})_{\mS}
	\end{align*}
	where the left side of the inequality is the lower bound on the query complexity under schema $\mR$ and the right side is the upper bound on the query complexity under schema $\mS$.
	The operator $>$ means that $A2$ will always ask asymptotically more queries under schema $\mR$ than under schema $\mS$.
	We have that $k_{\mR}$ and $k_{\mS}$ are bounded by $k$ and $m$ is the same for both schemas. We can also ignore $n$ as it is independent of the hypothesis space and the schemas. Therefore, by canceling out some terms, the previous inequality can be rewritten as
	$\Omega (k^a)_{\mR} > O(m (a-1) k^{2 + 3k} + (a-1) k^{2+k} )_{\mS}.$
	The first term in the upper bound dominates the second term, then we have
	$\Omega (k^a)_{\mR} > O(m (a-1) k^{2 + 3k} )_{\mS}.$
	Assuming that $m=1$, as in $\mL$, we get
	$\Omega (k^a)_{\mR} > O((a-1) k^{2 + 3k} )_{\mS}.$
	This inequality holds for sufficiently large $k$ and $a$.
\end{proof}

The lower bound of $A2$ is actually the Vapnik-Chevonenkis dimension (VC-Dim) of the hypothesis 
language that consists of function-free, first-order Horn expressions. 
Therefore, we have proven in Theorem~\ref{A2_schema_independent} that there are 
cases where the lower bound on the query complexity of {\it any} 
algorithm under this hypothesis language is greater than the upper bound on 
the query complexity of $A2$. This means that any algorithm that is as good as $A2$ 
(does not ask more queries than $A2$) is highly dependent on the schema details.


	\section{Experiments}
\label{section:empirical}

\subsection{Experimental Settings}
\label{section:empirical-settings}

\subsubsection{Datasets}
We use three datasets whose statistics are shown in Table~\ref{table:datasets}.

The {\bf HIV-Large} dataset contains information about 42,000 chemical compounds 
obtained from the National Cancer Institute's AIDS antiviral screen ({\it wiki.nci.nih.gov/display/NCIDTPdata}).
The initial schema contains relation {\it compound(comp,atm)}, which indicates that compound {\it comp} contains atom {\it atm}. 
The schema also has relations that indicate the chemical element that an atom represents, e.g., {\it element\_C(atm)}, as well as relations to indicate properties of each atom, e.g., {\it p2\_1(atm)}.
The schema represents a bond between two atoms by relation {\it bonds(bd,atm1,atm2)}, and it has a relation for each type of a bond, e.g., {\it bondType1(bd,t1)}.
The goal is to learn the target relation {\it hivActive(compound)}, which indicates that {\it compound} has anti-HIV activity. 
The original HIV dataset is stored in flat files and does not have any information about its constraints.
We explored the database for possible dependencies.
In particular, we have discovered that the INDs {\it bonds[bd] = bondType1[bd], 
bonds[bd] = bondType2[bd], bonds[bd] = bondType3[bd]} hold in the database.
We have used these dependencies to compose relations {\it bonds}, {\it bondType1}, {\it bondType2}, and {\it bondType3} into a single
relation {\it bonds} and create a schema in 4NF, named 4NF-1.
We also decompose relation {\it bonds} in the initial schema to relations {\it bondSource} and {\it bondTarget} to create another schema,
called 4NF-2.
The schemas and all INDs for this dataset are shown in Tables~\ref{table:hiv_schemas} and \ref{table:hiv-INDs}, respectively.
In the {\bf HIV-2K4K} dataset, we keep the same background knowledge, 
but reduce the number of examples to 2K positive and 4K negative examples.

\begin{table}
	\small
	\centering
	\begin{tabular} { c c c c c c }
		\hline
		Name & Schema & \#R & \#T & \#P & \#N \\
		\hline
		\multirow{3}{*}{HIV-Large} &  Initial & 80 & 14M & \multirow{3}{*}{5.8K} & \multirow{3}{*}{36.8K} \\ 
		&  4NF-1 & 77 & 7.8M &  &  \\ 
		&  4NF-2 & 81 & 16M &  &  \\ 
		\hline
		\multirow{4}{*}{UW-CSE} & Original & 9 & 1.8K & \multirow{4}{*}{102} & \multirow{4}{*}{204} \\ 
		&  4NF & 6 & 1.4K &  & \\ 
		&  Denormalized-1 & 5 & 1.3K &  &  \\ 
		&  Denormalized-2 & 4 & 1.3K &  &  \\ 
		\hline
		\multirow{3}{*}{IMDb} &  JMDB & 46 & 8.4M & \multirow{3}{*}{1.85K} & \multirow{3}{*}{3.6K} \\ 
		&  Stanford & 41 & 10.5M &  & \\
		&  Denormalized & 33 & 7.2M & & \\
	\end{tabular}
	\caption{ {\small Numbers of relations (\#R), tuples (\#T), positive examples (\#P), and negative examples (\#N) for each dataset.} }
	\label{table:datasets}
\end{table}
\begin{table}[h]
	\small
	\centering
	\begin{tabular} { l l l}
		\hline
		Initial & 4NF-1 & 4NF-2 \\
		\hline
		bonds(bd,atm1,atm2) & bonds(bd,atm1,atm2, & bSource(bd,atm1)\\
		bType1(bd,t1) & \ \ \ \ \ \ \ \ \ \ t1,t2,t3) & bTarget(bd,atm2) \\
		bType2(bd,t2) &  & bType1(bd,t1) \\
		bType3(bd,t3) &  & bType2(bd,t2) \\
		&  & bType3(bd,t3) \\
		&  & \\
		\hline
		\multicolumn{3}{c}{Common relations} \\
		\hline
		compound(comp, atm) & element\_C(atm) $...$ & element\_O(atm) \\
		p2\_0(atm) & p2\_1(atm) $...$ & p3(atm) \\
	\end{tabular}
	\caption{ {\small Schemas for the HIV-Large and HIV-2K4K datasets.}}
	\label{table:hiv_schemas}
\end{table}

\begin{table}[h]
	\small
	\centering
	\begin{tabular} { | l l|}
		\hline
		bonds[bd]=bType1[bd] & bonds[bd]=bType2[bd] \\
		bonds[bd]=bType3[bd] & \\
		\hline
		bonds[atm1]$\subseteq$compound[atm]  & bonds[atm2]$\subseteq$ compound[atm] \\   
		elem\_C[atm]$\subseteq$compound[atm] $\dots$ & elem\_O[atm]$\subseteq$ compound[atm] \\
		p2\_0[atm]$\subseteq$compound[atm]  $\ldots$ & p3[atm]$\subseteq$compound[atm]\\
		\hline
	\end{tabular}
	\caption{ {\small The INDs in the initial HIV dataset. The initial schema contains 80 INDs in total. }}
	\label{table:hiv-INDs}
\end{table}

\begin{table}[h]
	\small
	\centering
	\begin{tabular} { | l l|}
		\hline
		student[stud] $=$ inPhase[stud] & yearsInProg[stud] $\subseteq$ student[stud]  \\
		hasPosition[prof] $=$ professor[prof] & ta[stud] $\subseteq$ student[stud] \\
		ta[crs] $=$ taughtBy[crs] &  \\
		\hline
		taughtBy[prof] $=$ professor[prof] & student[stud] $\subseteq$ yearsInProg[stud] \\
		courseLevel[crs] $=$ taughtBy[crs] & \\
		\hline
		inPhase[stud] $\subseteq$ student[stud] & yearsInProg[stud] $\subseteq$ student[stud]  \\
		hasPosition[prof] $\subseteq$ professor[prof] & ta[stud] $\subseteq$ student[stud] \\
		taughtby[prof] $\subseteq$ professor[prof] & taughtby[crs] $\subseteq$ courseLevel[crs] \\
		\hline
	\end{tabular}
	\caption{ {\small Top: INDs in the original UW-CSE dataset. Middle: added INDs to have bijective transformations. Bottom: INDs that should hold according to the semantics of the database.}}
	\label{table:uwcse-INDs}
\end{table}

The {\bf UW-CSE} dataset ({\it alchemy.cs.washington.edu/data/uw-cse}) contains information about an academic department and 
has been used as a benchmark in the relational learning literature~\cite{Richardson:2006:MLN}.
The goal is to learn the target relation {\it advisedBy(stud,prof)}, as explained in Section~\ref{sec:introduction}.
The dataset comes with a set of constraints in form of first-order logic clauses that should hold over the dataset domain.
The INDs in these constraints are shown in Table~\ref{table:uwcse-INDs} (top).
The INDs in these constraints are {\it hasPosition[prof] = professor[prof], 
student[stud] = inPhase[stud],
ta[crs] = taughtBy[crs], 
yearsInProgram[stud] $\subseteq$ student[stud]}, and 
{\it ta[stud] $\subseteq$ student[stud]}. 
According to the original set of constraints, if one considers only the professors whose position is {\it Faculty},
the IND {\it taughtBy[prof] $=$ professor[prof]} holds.
If there are more INDs with equality in the schema, one can generate more schemas from the original UW-CSE schema using composition transformation.
To evaluate the effectiveness of algorithms over more varieties of schemas, 
we have considered only professors with position {\it Faculty} to use the IND {\it taughtBy[prof] = professor[prof]}.
For the same reasons, we also added the INDs {\it student[stud] $\subseteq$ yearsInProgram[stud]} and 
{\it courseLevel[crs] = taughtBy[crs]} to the schema. 
We enforce the aforementioned constraints by removing a small fraction of tuples, 159 tuples, from the original dataset.
We iteratively compose the original schema to four different schemas, two of which are shown in Table~\ref{table:uwcse}.
We compose {\it courseLevel} and {\it taughtBy} relations in 4NF schema to create the a more denormalized schema, named Denormalized-1, and 
compose {\it courseLevel}, {\it taughtBy}, and {\it professor} in 4NF schema to generate the fourth schema, named Denormalized-2.

\begin{table}[h]
	\small
	\centering
	\begin{tabular} { l l }
		\hline
		JMDB & Stanford  \\
		\hline
		movie(id,title,year) &  movie(id,title,year,genreid, \\
		movies2genre(id,genreid) & \hspace{0.5cm} colorid,prodcompid,\\
		movies2color(id,colorid) & \hspace{0.5cm} directorid,producerid)\\
		movies2director(id,directorid) & \\
		movies2producer(id,producerid) & \\
		movies2prodcomp(id,prodcompid) & \\
		\hline
		\multicolumn{2}{c}{Common relations} \\
		\hline
		language(id,language) & plot(id,plot) \\
		country(id,country) & color(id,color)\\
		business(id,text) & altversion(id,version) \\
		runningtime(id,times) & prodcompany(id,name) \\
		actor(id,name,sex) & editor(id,name)\\
		director(id,name) & producer(id,name)  \\
		writer(id,name) & akaname(name,akaname) \\
		akatitle(id,langid,title) & cinematgr(id,name) \\
		biography(id,name,bio) & movies2misc(id,miscid) \\
		composer(id,name) & costdesigner(id,name) \\
		distributor(id,name) & rating(id,rank,votes)\\
		genre(id,genre) & misc(id,name)  \\
		mpaarating(id,text) & technical(id,text) \\
		proddesinger(id,name) & releasedate(id,countryid,date) \\
		movies2actor(id,actorid,character) & movies2editor(id,editorid) \\
		movies2writer(id,writerid) & movies2cinematgr(id,cinamtid) \\
		movies2composer(id,composerid) & movies2costdes(id,costdesid)\\
		movies2language(id,langid) & certificate(id,countryid,cert)\\
		movies2proddes(id,proddesid) & movies2country(id,countryid) \\
	\end{tabular}
	\caption{{\small JMDB and Stanford schemas for the IMDb dataset. Relations in bottom are contained in both schemas. }}
	\label{table:imdb_schemas_1}
\end{table}

\begin{table}[h]
	\small
	\centering
	\begin{tabular} { l l }
		\hline
		\multicolumn{2}{c}{Denormalized} \\
		\hline
		movie(id,title,year) &  language(id,language) \\
		movies2actor(id,actorid,name, & plot(id,plot)\\
		\hspace{1.5cm} character,sex) & business(id,text)\\
		movies2color(id,colorid,color) & altversion(id,version)\\
		movies2X(id,Xid,name) s.t. & runningtime(id,times)\\
		X$=\{\textnormal{writer,editor,composer,}$ & prodcompany(id,name)\\
		\hspace{0.2cm} $\textnormal{cinematgr,costdes,proddes,}$ & country(id,country)\\
		\hspace{0.2cm}$\textnormal{director,producer,misc}\}$ &  akaname(name,akaname)\\
		akatitle(id,langid,title) & biography(id,name,bio) \\
		distributor(id,name) & rating(id,rank,votes)\\
		genre(id,genre) & releasedate(id,countryid,date) \\
		movies2language(id,langid) & certificate(id,countryid,cert)\\
		mpaarating(id,text) & technical(id,text) \\
		movies2country(id,countryid) & \\
	\end{tabular}
	\caption{ {\small Denormalized schema for the IMDb dataset. }}
	\label{table:imdb_schemas_2}
\end{table}

\begin{table}[h]
	\small
	\centering
	\begin{tabular} { |  c|}
		\hline
		movies2X[id] $=$ movie[id] \\ 
		s.t. X$=\{\textnormal{genre, color, prodcompany, producer, director}\}$\\
		movies2Y[Yid] $=$ Y[id] \\
		s.t. Y$=\{\textnormal{actor, cinematagr, color, composer, costdes, director,}$\\
		$\textnormal{editor, misc, proddes, producer, writer}\}$\\
		\hline
		Z[id] $\subseteq$ movie[id]\\
		s.t. Z=$\{\textnormal{business, runningtime, altversion, certificate,}$\\ 
		$\textnormal{plot, rating, akatitle, distributor, releasedate,}$\\ 
		$\textnormal{technical, movies2actor, movies2country, movies2composer,}$\\
		$\textnormal{movies2writer, movies2costdes, movies2misc, movies2editor,}$\\
		$\textnormal{movies2cinematgr, movies2language, movies2proddes}\}$\\
		certificate[countryid] $\subseteq$ country[countryid] \\ 
		releasedate[countryid] $\subseteq$ country[countryid] \\
		akatitle[langid] $\subseteq$ language[langid] \\ 
		movies2country[countryid] $\subseteq$ country[countryid] \\
		movies2language[langid] $\subseteq$ language[langid] \\ 
		movies2genre[genreid] $\subseteq$ genre[genreid] \\
		movies2prodcompany[prdcompid] $\subseteq$ prodcompany[prdcompid] \\
		\hline
	\end{tabular}
	\caption{ {\small INDs in IMDB dataset.}}
	\label{table:imdb-INDs}
\end{table}

The {\bf IMDb} ({\it imdb.com}) dataset contains information about movies.
We learn the target relation {\it dramaDirector(director)}, which indicates that {\it director} has directed a drama movie.
JMDB ({\it jmdb.de}) provides a relational database of IMDb data under a 4NF schema.
We create a subset of JMDB database by selecting the movies produced after year 2000 and their related entities, 
e.g., actors, directors, producers. 
The relationships between relation {\it movie(id,title,year)} and its related relations, e.g., {\it director(id,name)}, 
are stored in relations {\it movies2X} where X is the name of the related entity set, e.g., {\it movies2director(id,directorid)}.
The resulting database has 11 INDs with equality in form of {\it movies2X[Xid] =X[id]}, e.g., \\
{\it movies2director[directorid] = director[id]}.  
To test over more transformations, we have changed some regular INDs 
in the database in form of {\it movies2X[id] $\subseteq$ movie[id]} to 
{\it movies2X[id] = movie[id]}
where X is {\it genre}, {\it color}, {\it prodcompany}, {\it producer}, and  {\it director} by removing some tuples from the database. 
We use the first set of 11 INDs with equality to compose 11 pairs of relations in JMDB schema, e.g., composing 
{\it movies2director(id, directorid)} and {\it director(id, name)} into {\it movies2director(id, directorid, name)},  
to create a new schema, called Denormalized.
We use the second set of INDs with equality to compose 5 relations in JMDB schema, e.g., {\it movies2genre}, into {\it movie} relation
and create a schema called Stanford that follows a structure similar to the one used in the Stanford Movie DB ({\it infolab.stanford.edu/pub/movies}). 
We explored our JMDB database to find other INDs, which are listed in Table~\ref{table:imdb-INDs}.
The three schemas and the full list of INDs in IMDb data are shown in Tables~\ref{table:imdb_schemas_1}, \ref{table:imdb_schemas_2} and \ref{table:imdb-INDs}. 
In the UW-CSE and IMDb datasets, we generate negative examples by using the closed-world assumption, 
and then sample to obtain twice as many negative examples as positive examples.

\subsubsection{Algorithms}
We compare Castor to three relational learning systems: 
FOIL~\cite{Quinlan:FOIL}, Aleph~\cite{aleph}, and GILPS~\cite{progolem}.
{\bf FOIL} system implements FOIL algorithms but does not scale to medium and large datasets.
Therefore, we also emulate FOIL using Aleph by forcing Aleph to follow a greedy strategy and call it {\bf Aleph-FOIL}.
Aleph is a well known ILP system that implements Progol by its default setting~\cite{progol}.
To differentiate the two variations of Aleph used in our experiment, 
we call the default implementation of Aleph {\bf Aleph-Progol}. 
GILPS implements {\bf ProGolem}, which is a bottom-up algorithm.
We cannot compare Castor to QuickFOIL~\cite{QuickFOIL} as it is not publicly available. 

There are far fewer query-based relational learning 
systems available than the ones that use samples for learning.
To empirically evaluate the schema independence of query-based 
learning methods, we use the {\bf LogAn-H} system~\cite{Arias:2007}, which is an implementation 
of the A2 algorithm~\cite{Khardon:1999}. 
We call the learning algorithms that use 
batches of training samples, 
e.g., FOIL and ProGolem, sample-based algorithms to 
distinguish them from query-based algorithms in this section. 

Machine learning algorithms usually require parameter tuning to run them successfully.
We try to use the default parameter configuration for all systems, except when needed.
Because we use noisy datasets, we must allow the algorithms to learn clauses that cover some negative examples.
To limit the number of negative examples covered by any learned clause, we require 
that the ratio of positive to negative examples covered by a clause (precision) is at least 2 to 1. 
That is, the number of positive examples examples covered by a clause must be two times greater than or equal to the number of negative examples covered by the clause.
In FOIL, this value is set with the {\it aaccur} parameter; in Aleph it is set with the {\it minacc} parameter; in ProGolem and Castor it is set with the {\it minprec} parameter.
In FOIL, the only settings that we modify is {\it aaccur=0.67}.
In Aleph, the settings that we modify are {\it minacc=0.67}, {\it minpos=2}, {\it noise=inf}, and {\it openlist=1} (only for Aleph-FOIL).
In Castor and ProGolem, the settings are {\it minprec=0.67}, {\it noise=1}, {\it minpos=2}, and {\it sample=1}, {\it beamwidth=1} for HIV-Large, HIV-2K4K, and IMDb, and {\it sample=20}, {\it beamwidth=3} for UW-CSE.
In the IMDb dataset, we also restrict the number of literals with the same relation symbol added to a ground bottom clause in one iteration of the bottom clause construction algorithm.
We set this value to 10.
If this value is unrestricted, a bottom clause may contain hundreds or thousands of literals with the same relation symbol (one for each tuple).

Top-down algorithms contain the parameter {\it clauselength}, which sets an upper bound on the number of literals in a clause.
The default value for this parameter in Aleph is 4.
Over HIV-Large and HIV-2K4K, the definition for the target relation must contain long clauses.
With {\it clauselength} $= 4$, Aleph-FOIL and Aleph-Progol do not learn any clause.
Therefore, we set this parameter to have values of 10 and 15.

\subsubsection{Evaluation Metrics}
We compare the quality of the leaned definitions using the metrics of {\it precision} and {\it recall}.
Let the set of {\it true positives} for a definition be the set of positive examples in the testing data that are covered by the definition.
The precision of a definition is the proportion of its true positives over all examples covered by the definition.
The recall of a definition is the number of its true positives divided by the total number of positive examples in the testing data.
Precision and recall are between 0 and 1, where an ideal definition delivers both precision and recall of 1.
Similar to other machine learning tasks, 
it is not often possible to learn an ideal definition for a target concept due to various reasons,
such as the hardness of the target concept or the lack of sufficient amount of training data. 
In these situations, the values of reasonable precision and recall for a definition depend on the underlying applications, 
e.g., 5\% improvement in precision may not be important in a financial application but vital in a medical application.
Nevertheless, definitions with higher precision and/or recall are generally more desirable \cite{Quinlan:FOIL,progolem,aleph}.
We perform 5-fold cross validation for UW-CSE and 
10-fold cross validation for HIV and IMDb datasets.
We evaluate precision, recall, and running times, showing the average over the cross validation.

\subsubsection{Other Settings}
Experiments were run on a server containing 32 2.6GHz Intel Xeon E5-2640 processors, running CentOS Linux 7.2 with 50GB of main memory.

\subsection{Sample-based Algorithms}

Castor is schema independent over all datasets and 
delivers equal precision and recall across all schemas of each dataset in our experiments.
However, other algorithms are schema dependent.\\

{\bf HIV datasets.}
Aleph-FOIL, Aleph-Progol and Castor are the only algorithms that scale to the HIV-2K4K dataset.
Aleph-FOIL and Castor also scale to the HIV-Large dataset.
The definitions learned by Aleph-FOIL and Aleph-Progol over different schemas are not equivalent
as shown by their precision and recall values across schemas in  Table~\ref{table:hiv}. 
Different schemas cause Aleph-FOIL and Aleph-Progol to explore different regions of the hypothesis space.
Aleph-FOIL and Aleph-Progol are not able to find any definition over the 4NF-2 schema of HIV-Large and HIV-2K4K datasets.
The reason is that any good clause must contain information about bonds. 
In the 4NF-2 schema, this information is represented by two relations, {\it bondSource} and {\it bondTarget}, and three more to indicate their types. 
With a top-down search, these algorithms are not able to find a clause that contains these relations.
Aleph-FOIL terminates without learning anything and Aleph-Progol does not terminate after 75 hours.

Aleph-Progol does not terminate after 75 hours over the 4NF-2 schema of HIV-2K4K.
FOIL crashes on both HIV datasets.
ProGolem does not learn anything after 5 days running, even on smaller subsets of the HIV dataset.\\

{\bf UW-CSE dataset.}
As shown in Table~\ref{tableuwcseres67}, all algorithms except for Castor are schema dependent and learn non-equivalent definitions over 
different schemas of UW-CSE.
As this dataset is smaller than HIV and IMDb datasets, it has a relatively smaller hypothesis space.
Hence, the degree of schema dependence for these algorithms over this dataset is generally lower than other datasets.
This is reflected in their precision and recall, which are not significantly different across schemas.

Over denormalized schemas, Aleph-FOIL learns definitions consisting of many clauses, each covering a few examples. 
This results in low generalization, hence very low precision and recall.
On the other hand, over the Original schema, it learns definitions consisting of a lower number of clauses, each covering a greater number of examples.
Note that Aleph-FOIL does not exactly emulate FOIL. 
FOIL uses a different evaluation function and explores an unrestricted hypothesis space. 
Therefore, FOIL does not suffer from the same problems as Aleph-FOIL.
However, it is less effective than other algorithms.
Castor's effectiveness is comparable to Aleph-Progol and ProGolem over the Original and 4NF schemas. 
Nevertheless, Aleph-Progol and ProGolem perform worse on other schemas. 
On the other hand, Castor is effective over all schemas.\\


\begin{table}
	\begin{center}
		{\small 
			\centering
			\begin{tabular} {c| c|c|c|c}
				\hline
				\multicolumn{5}{c}{HIV-Large} \\
				\hline
				Algorithm & Metric & Initial & 4NF-1 & 4NF-2 \\
				\hline			
				Aleph-FOIL & Precision	& 0.58	& 0.72 & 0 	\\
				($\mathit{clauselength} = 10$) & Recall	& 0.42	& 0.91 & 0	 	\\
				& Time (hours)	& 3	& 0.9 & 6 	\\
				\hline
				Aleph-FOIL & Precision	& 0.68	& 0.68 & 0  	\\
				($\mathit{clauselength} = 15$) & Recall	& 0.41	& 0.85 & 0	 	\\
				& Time (hours)	& 11.7	&  3.7 & 47	\\
				\hline
				\multirow{3}{*}{Castor} & Precision	& 0.81 & 0.81 & 0.81 \\
				& Recall	& 0.85 & 0.85 & 0.85 \\
				& Time (hours)	& 3.5 & 1.9 & 56 \\
				\hline
				\multicolumn{5}{c}{HIV-2K4K} \\
				\hline		
				Aleph-FOIL & Precision	& 0.72	&  0.78 & 0	\\
				($\mathit{clauselength} = 10$) & Recall	& 0.69	& 0.81 & 0	 	\\
				& Time (min)	& 6.2	& 7.9 & 20.6	 	\\
				\hline
				Aleph-FOIL & Precision	& 0.70	& 0.78 & 0  	\\
				($\mathit{clauselength} = 15$) & Recall	& 0.79	& 0.89	& 0	\\
				& Time (min)	& 6.72	& 7.07	& 122.2	\\
				\hline
				Aleph-Progol & Precision	& 0.70	& 0.79 & -	\\
				($\mathit{clauselength} = 10$) & Recall	& 0.85	& 0.90	& -	\\
				& Time (min)	& 58.5	& 72.2	& $>$ 75 h	\\
				\hline
				Aleph-Progol & Precision	& 0.72	&  0.75 & -	\\
				($\mathit{clauselength} = 15$) & Recall	& 0.89	& 0.87 & -	\\
				& Time (min)	& 155.51	& 13.56	& $>$ 75 h	\\
				\hline
				\multirow{3}{*}{Castor} & Precision	& 0.80 & 0.80 & 0.80\\
				& Recall	& 0.87 & 0.87 & 0.87 \\
				& Time (min)	& 15.1 & 6.5 & 335.5 \\
			\end{tabular}
		}
		\caption{{\small Results of learning relations over HIV-Large and HIV-2K4K data.}}
		\label{table:hiv}
	\end{center}
\end{table}

\begin{table}[t]
	\begin{center}
		{\small 
			\centering
			\begin{tabular} { c | c | c | c | c | c }
				\hline
				Algorithm & Metric & Original & 4NF & Denorm-1 & Denorm-2  \\
				\hline			
				\multirow{3}{*}{FOIL} & Precision	& 0.84	& 0.79	& 0.77 & 0.85 \\
				& Recall	& 0.35	& 0.36	 & 0.42 & 0.47 \\
				& Time (s)	& 18.7	& 20.84 & 30.72 & 30.64	\\
				\hline
				\multirow{3}{*}{Aleph-FOIL} & Precision	& 0.78	& 0.50 & 0.36 & 0.19 \\
				& Recall	& 0.17	& 0.18 &	0.13 & 0.11 \\
				& Time (s)	& 3.5	& 4.3 & 14.8 & 398.1 \\
				\hline
				\multirow{3}{*}{Aleph-Progol} & Precision	& 0.95	& 0.97 &	0.98 & 0.55 \\
				& Recall	& 0.54	& 0.45	& 0.36 & 0.29 \\
				& Time (s)	& 9.7	& 13.2	& 27.9 & 334.8 \\
				\hline
				\multirow{3}{*}{ProGolem} & Precision	& 0.95	& 0.95	& 0.80 & 0.82\\
				& Recall	& 0.54	& 0.54	& 0.48 & 0.48 \\
				& Time (s)	& 24.4	& 28.8 &	26.7 & 54.1 \\
				\hline
				\multirow{3}{*}{Castor} & Precision	& 0.93 & 0.93 & 0.93 & 0.93 \\
				& Recall	& 0.54 	& 0.54	& 0.54 & 0.54 \\
				& Time (s)	& 7.2 & 7.4 & 7.9 & 12.4 \\
			\end{tabular}
		}
		\caption{{\small Results of learning relations over UW-CSE data.}}
		\label{tableuwcseres67}
	\end{center}
\end{table}

\begin{table}[t]
	\begin{center}
		{\small 
			\centering
			\begin{tabular} { c | c | c | c | c }
				\hline
				Algorithm & Metric & JMDB & Stanford & Denormalized \\
				\hline			
				\multirow{3}{*}{Aleph-FOIL} & Precision	& 0.66	& 0.92 & 0.67	 \\
				& Recall	& 0.44	& 	1 & 0.45 \\
				& Time (min)	& 	6.4 & 1,229 & 476.4	\\
				\hline
				\multirow{3}{*}{Aleph-Progol} & Precision	& 	0.66 & 1 & 0.69	 \\
				& Recall	& 	0.47& 1 & 0.52	 \\
				& Time (min)	& 	312.9 & 1,248 & 937.4	\\
				\hline
				\multirow{3}{*}{Castor} & Precision	& 1 & 1 & 1  \\
				& Recall	& 1 & 1 & 1 \\
				& Time (min)	& 15.14 & 108.15 & 32.4 \\
			\end{tabular}
		}
		\caption{{\small Results of learning relations over IMDb data.}}
		\label{tableimdbres67}
	\end{center}
\end{table}

{\bf IMDb dataset.}
The target relation for the IMDb dataset has an exact Datalog definition given the background knowledge and training examples.
Castor finds this definition over all schemas and obtains precision and recall of 1, as shown in Table~\ref{tableimdbres67}.
Aleph-FOIL fails to find this definition over all schemas.
Aleph-Progol finds this definition only over the Stanford schema.
The definitions learned by Aleph-FOIL and Aleph-Progol over different schemas are largely different. \\
{\bf Relationship between style of design and effectiveness.}
Our results show that there is not any single style of design, e.g., 4NF, on which all algorithms, except for Castor, 
are effective over all datasets. 
Generally, the style of design on which a relational learning algorithm delivers its most effective results varies based on 
the metric of effectiveness, the dataset, and the algorithm.
For example, Aleph-Progol delivers its highest precision over a denormalized schema, Denormalized-1, for  
UW-CSE, but its highest recall over the original schema, which is more normalized than 4NF. 
Aleph-Progol also delivers its lowest precision on UW-CSE data over another denormalized schema, Denormalized-2, for this dataset.
Hence, it is generally hard to find a straightforward relationship between the style of 
design and the precision or recall of an algorithm over a given dataset.
Furthermore, each algorithm prefers a different style of design over each dataset. 
For example, Aleph-Progol has higher overall precision and recall
on the most normalized schema, original schema, for UW-CSE. 
But, it delivers its highest overall precision and recall over the most denormalized schema, Stanford, for IMDb.
Finally, different algorithms prefer distinct styles of design over the same dataset.
For example,  FOIL delivers both its highest precision and highest recall 
over a denormalized schema for UW-CSE data, Denormalized-2, over which Aleph-Progol delivers both its lowest precision and lowest recall.
Over the same database, ProGolem achieves both its highest precision and highest recall for the most normalized schema, i.e., original schema. \\

{\bf Efficiency.}
Besides being schema independent, Castor offers the best trade-off between effectiveness and efficiency.
Generally, Aleph-FOIL is more efficient than Castor, but less effective.
Aleph-Progol is usually effective, but becomes very inefficient as the size of data grows.
FOIL and ProGolem only scale to small datasets.

Aleph-FOIL and Castor are the only algorithms that scale to the HIV-Large dataset.
Aleph-FOIL with $\mathit{clauselength} = 10$ is more efficient than Castor.
However, when {\it clauselength} is set to 15, it becomes less efficient, as shown in Tables~\ref{table:hiv}.
Aleph-FOIL with both {\it clauselength} = $10$ and $15$ is also faster than Castor over the HIV-2K4K dataset.
In general, top-down algorithms that follow greedy search strategies are expected to be more efficient than bottom-up algorithms.
Top-down algorithms have a search bias for shorter clauses, which are cheaper to compute.
They usually limit the maximum length of the clauses to be learned.
Further, algorithms that follow greedy search strategies can be more efficient.
This is exploited by related work that focuses on efficiency~\cite{QuickFOIL,join:crossmine,amie-plus}.
However, as the maximum clause length is increased, the hypothesis space grows, and these algorithms become less efficient.
Top-down algorithms that do not follow a greedy search strategy, such as Progol, are generally not efficient. 
This is reflected in our empirical studies, where Aleph-Progol did not scale to the HIV-Large dataset, and is the slowest algorithm on the HIV-2K4K dataset.

Castor is able to scale to large databases such as HIV-Large and HIV-2K4K because of the optimizations explained in Section~\ref{section:implementation}.
By reusing information about previous coverage tests, Castor reduces the number of coverage tests on new clauses.
This is particularly useful on large databases with complex schemas, such as the HIV datasets, where generated clauses are large and expensive to evaluate.
Parallelization also helps Castor on reducing the time spent on coverage testing. 
For these experiments, Castor parallelized coverage testing by using 32 threads.
Finally, minimization helps in reducing the size of clauses.
For instance, over both of HIV datasets, Castor reduces the size of bottom-clauses over the Initial schema by $19\%$, over the 4NF-1 schema by $13\%$, and over the 4NF-2 schema by $18\%$, on average.
Castor removes redundant literals from the bottom-clause, which results in reducing the search space and the cost of performing coverage tests.
Note that the running time of all algorithms increases significantly over the 4NF-2 schema of the HIV-Large and HIV-2K4K datasets. 
As the {\it bond} relation is decomposed into {\it bondSource} and {\it bondTarget} in this schema,
the number of tuples to represent bonds is doubled compared to the Initial schema.
Therefore, algorithms must explore clauses with a large number of literals, hundreds, whose coverage testings take a very long time.
We plan to optimize the coverage testing engine of Castor to efficiently process such datasets.

The efficiency of Castor is comparable to the efficiency of Aleph-FOIL and Aleph-Progol over the Original and 4NF schemas of the UW-CSE dataset. 
The running time of Aleph-FOIL and Aleph-Progol is heavily impacted over the Denormalized-2 schema, as shown in Table~\ref{tableuwcseres67}. 
Castor is efficient over all schemas of this dataset.
UW-CSE is the only dataset for which FOIL and ProGolem scale. However, in general, they are less efficient.

Castor is significantly more efficient and effective than Aleph-FOIL and Aleph-Progol on the IMDb dataset, as shown in Table~\ref{tableimdbres67}.
In general, top-down algorithms are efficient 
if they take the correct first steps when searching for the definition.
In this case, Aleph-FOIL and Aleph-Progol (over two schemas) take the wrong steps and focus on a section of the hypothesis space that does not contain the correct definition.\\

\begin{table}
	\begin{center}
		{\small 
			\centering
			\begin{tabular} {c|c|c|c|c}
				\hline
				\multicolumn{5}{c}{HIV-2K4K} \\
				\hline			
				Metric & Initial & \multicolumn{2}{c|}{4NF-1} & 4NF-2  \\
				\hline
				Precision	& 0.77	& \multicolumn{2}{c|}{0.79} & 0.73 	\\
				Recall	& 0.63	& \multicolumn{2}{c|}{0.87} & 0.76 	\\
				Time (min)	& 27	&  \multicolumn{2}{c|}{14.8} & 576 \\
				
				\hline
				\multicolumn{5}{c}{UW-CSE} \\
				\hline			
				Metric & Original & 4NF & Denorm-1 & Denorm-2  \\
				\hline
				Precision	& 0.93	& 0.93 & 0.93 & 0.93	\\
				Recall	& 0.54	& 0.54 & 0.54	& 0.54	\\
				Time (s)	& 8	&  8.9 & 9.1	& 13.3 \\
				
				\hline
				\multicolumn{5}{c}{IMDb} \\
				\hline			
				Metric & JMDB & \multicolumn{2}{c|}{Stanford} & Denormalized  \\
				\hline
				Precision	& 1	& \multicolumn{2}{c|}{0.98} & 1 	\\
				Recall	& 1	& \multicolumn{2}{c|}{0.84} & 1 	\\
				Time (min)	& 7.3	&  \multicolumn{2}{c|}{90.8} & 8.1	 \\
			\end{tabular}
		}
		\caption{{\small Results of Castor learning relations over HIV-2K4K, UW-CSE and IMDb data using only INDs in the form of subset.}}
		\label{table:beyondcomposition}
	\end{center}
\end{table}

{\bf General decomposition/ composition.}
As it is explained in Section~\ref{sec:generalDecomposition}, 
there are two methods to achieve robustness over the schema variations
created by the INDs in general forms.
One can use a preprocessing step to check whether the IND holds in the 
form of equality over the available instance. Then, one 
can apply the original Castor algorithm and
achieve complete schema independence. The empirical results of this method are
exactly the same as the ones of the original Castor algorithm with the
overhead of its preprocessing step.
Another method is to use the INDs in general form directly without any preprocessing.
We empirically evaluate the robustness of the latter method in this section. 
To explore general decomposition/ compositions of 
HIV, UW-CSE, and IMDb, we restore the INDs with equality that we have enforced on their schemas to their original forms. 
For instance, we restore the enforced INDs with equality {\it movies2X[id] = movie[id]} in IMDb schemas to  
{\it movies2X[id] $\subseteq$ movie[id]} in IMDb schemas. 
We also modify the INDs with equality that are originally found in these datasets to INDs in form of foreign key 
to primary key referential integrities in their schemas. 
For example, we have changed INDs {\it movies2X[Xid] = X[id]} to {\it movies2X[Xid] $\subseteq$ X[id]} over IMDb schemas. 
Hence, the transformations explained in Section~\ref{section:empirical-settings} for these datasets are 
general decomposition/ composition and not bijective. 
We run the extended version of Castor from Section~\ref{sec:generalDecomposition} using the aforementioned INDs and all other regular INDs in each schema.
The extension of Castor gets the same results as in Table~\ref{tableuwcseres67} 
over UW-CSE and is schema independent. It is also robust and
delivers the same results as in Table~\ref{tableimdbres67} for JMDB and Denormalized schemas 
of IMDb. But, it returns precision of 0.98 and recall of 0.84 over the database with 
Stanford schema. Overall, it is more effective and schema independent than other algorithms over IMDb.
However, the results of the extension of Castor vary with the schema over the HIV-2K4K dataset: it delivers precision of 
0.77,  0.79, and 0.73 and recall of 0.63, 
0.87, and 0.76 over the Initial, 4NF-1, and 4NF-2 schemas, respectively. 
This is because it cannot access some tuples in the bottom-clause construction in these databases as explained in Section~\ref{sec:generalDecomposition}. 
Its precisions are equal or higher than the those of Aleph-FOIL and Aleph-Progol over all schemas and its recall is higher than that of Aleph-FOIL and Aleph-Progol in 4NF-2 schema.
But, its recall is lower than the recall of Aleph-FOIL and Aleph-Progol over the Initial and Aleph-Progol over 4NF-1 schemas.
Table~\ref{table:beyondcomposition} shows the results of Castor learning relations over the HIV-2K4K, UW-CSE and IMDb datasets, using only INDs in the form of subset.
For HIV-2K4K, it uses the INDs in Table~\ref{table:hiv-INDs} (bottom).
For UW-CSE, it uses the INDs in Table~\ref{table:uwcse-INDs} (bottom).
For IMDb, it uses the INDs in Table~\ref{table:imdb-INDs} (bottom).

\subsection{Impact of Castor Design Choices}

We evaluate the impact of parallelization and the use of stored procedures on Castor's running time.
There are some variations in the running times of Castor compared to the experiments in the previous section. This is because we run experiments again, and the running times may fluctuate.

\begin{figure}
	\includegraphics[width=0.5\linewidth]{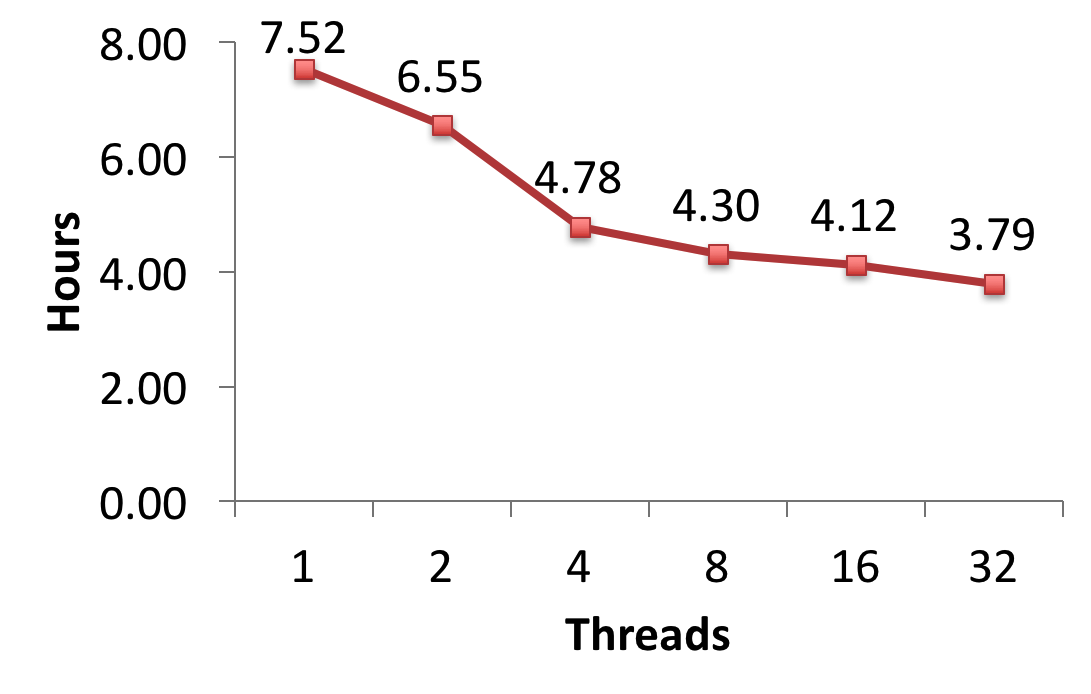}\hfill
	\includegraphics[width=0.5\linewidth]{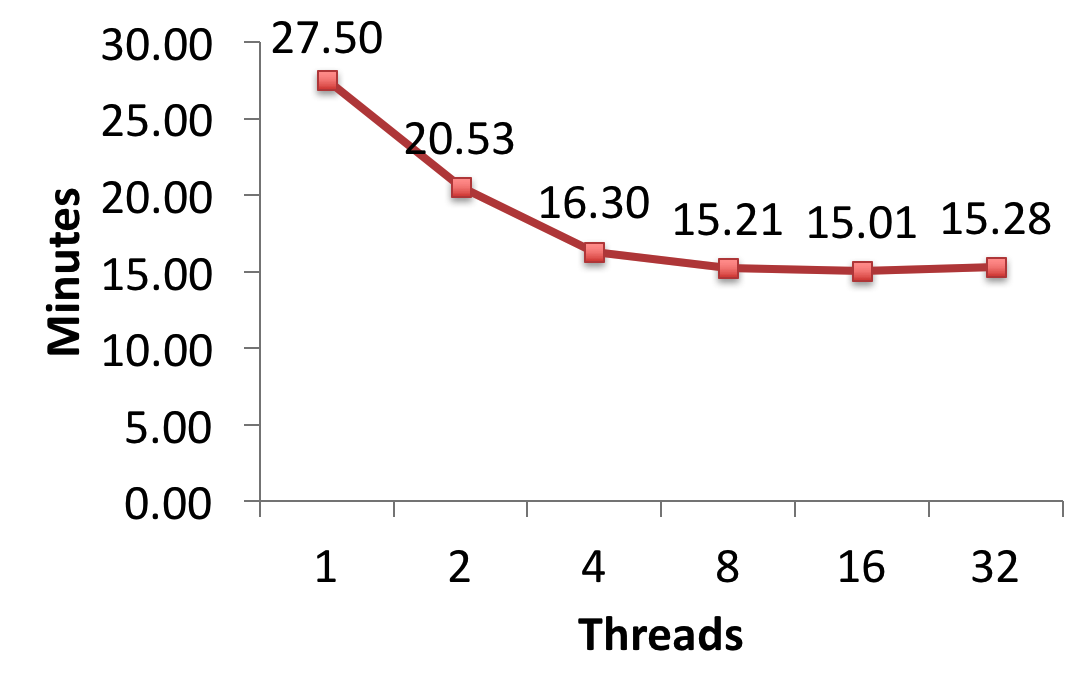}\hfill
	\includegraphics[width=0.5\linewidth]{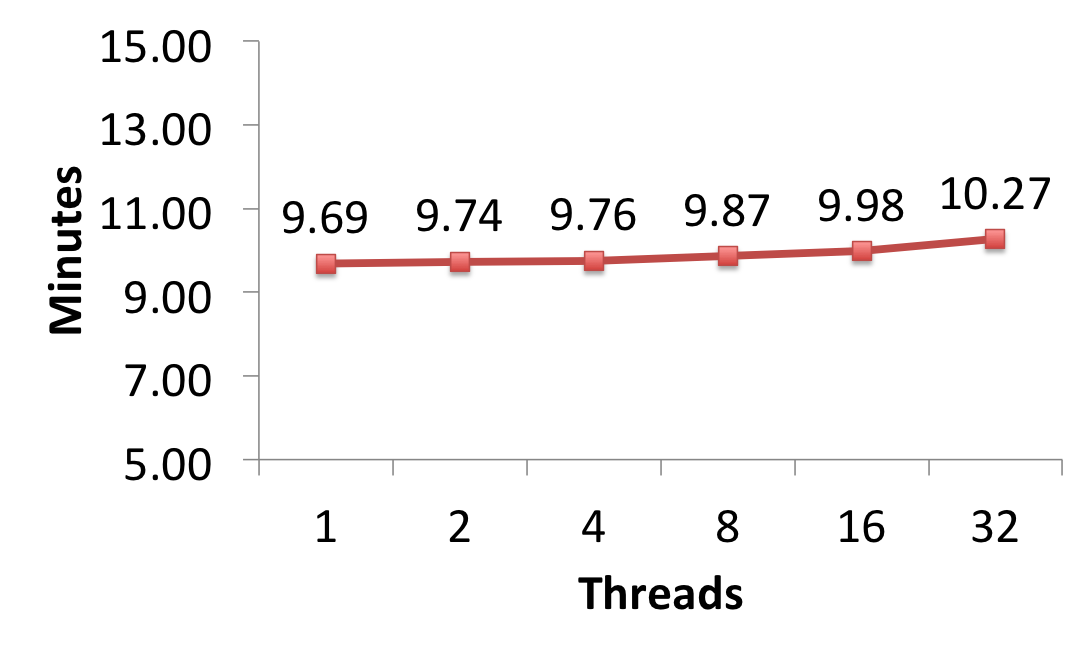}\hfill
	\caption{{\small Impact of parallelization on Castor's running time over the HIV-Large (top-left), HIV-2K4K (top-right), and IMDb datasets (bottom).}}
	\label{figure:castor-parallelization}
\end{figure}

{\bf Impact of parallelization.}
Castor performs coverage tests in parallel to improve its running time.
Figure~\ref{figure:castor-parallelization} shows the impact of parallelization on Castor's running time over the HIV-Large (Initial schema), HIV-2K4K (Initial schema) and IMDb (JMDB schema) datasets.
Over both HIV-Large and HIV-2K4K datasets, Castor benefits from parallelization.
Over the HIV-Large dataset, the best performance is obtained by using 32 threads, which reduces the running time by half compared to using 1 thread.
Over the HIV-2K4K dataset, the running time also reduces significantly with parallelization and the best performance is obtained with 16 threads.
Over the IMDb dataset, there is no benefit in using parallelization. This is because Castor does not need to perform many coverage tests, as it is able to find the perfect definition very quickly. 
In this case, most of Castor's running time is spent in creating the ground bottom-clauses, as explain in Section~\ref{section:implementation}.
Because the UW-CSE dataset is very small, there is no need for parallelization.
Notice that sequential Castor (1 thread) is more efficient than Aleph-FOIL with {\it clauselength = 15} over the HIV-Large dataset and more efficient than Aleph-Progol over the HIV-2K4K and IMDb datasets.
This shows that besides parallelization, the techniques explained in Section~\ref{section:implementation} allow Castor to run efficiently.

{\bf Impact of using stored procedures.}
The bottom-clause construction algorithm is used to generate ground bottom-clauses, used to evaluate clauses, as well as in Castor's {\it LearnClause} procedure.
As mentioned in Section~\ref{section:implementation}, we implement the bottom-clause construction algorithm inside a stored procedure.
To evaluate the benefit of using stored procedures, we also implement a version of Castor that does not use stored procedures. 
Table~\ref{table:castor-storedprocedures} shows the running time of the versions of Castor with and without stored procedures over the HIV-Large (Initial schema), HIV-2K4K (Initial schema) and IMDb (JMDB schema) datasets.
The version of Castor that uses stored procedures obtains between 1.25x and 1.9x speedup over the version that does not use stored procedures.


\begin{table}[t]
	\begin{center}
		{\small 
			\centering
			\begin{tabular} { c | c | c }
				\hline
				Dataset & With stored procedures & Without stored procedures \\
				\hline
				HIV-Large & 3.79h & 4.75h \\
				HIV-2K4K & 15.28m & 25.23m \\
				IMDb & 10.27m & 19.49m \\
			\end{tabular}
		}
		\caption{{\small Impact of stored procedures on Castor's running time over the HIV-Large, HIV-2K4K, and IMDb datasets}}
		\label{table:castor-storedprocedures}
	\end{center}
\end{table}

\subsection{Query-based Algorithms}
We used the interactive 
algorithm with automatic user mode in the LogAn-H system. In this mode, the system is told 
the Horn definition to be learned, so that it can act as an oracle. Then the algorithm's queries are answered automatically until it learns the exact definition.
We performed experiments over the UW-CSE dataset.
Query-based algorithms do not scale to larger datasets.
We generated random Horn definitions over the Denormalized-2 schema of the UW-CSE dataset. 
Each definition contains one or more clauses.
The definition generator has a parameter to indicate the number of variables in each clause.
To generate the head of each clause, we created a new relation of 
random arity, where the minimum arity is 1 and the maximum arity is 
the maximum arity of the relations in the Denormalized-2 schema. 
The body of each clause can be of any length as long as the number 
of variables in the clause is equal to the specified parameter and 
all variables appearing in the head relation also appear in any relation in the body. 
The body of the clause is composed of randomly chosen relations, 
where each relation can be the head relation (allowing for recursive clauses) 
or any relation in the input schema. 
Head and body relations 
are populated with variables, where each variable is randomly 
chosen to be a new (until reaching the input number of variables) 
or already used variable. 
Clauses cannot contain function or constant symbols. 

After generating each random Horn definitions over the Denormalized-2 schema, we transformed these expressions to the Denormalized-1, 4NF and Original schemas by simply doing vertical 
decomposition to each of the clauses in a definition. 
We varied the number of clauses in a definition to be between 1 and 5, each containing between 4 and 8 variables. 
Therefore, we generated 50 random definitions for each setting.
The A2 algorithm takes as input the target definition and the schema. 
We ran the LogAn-H system and recorded the number of queries required to learn each definition under each schema. 
In these experiments, we report the average query complexity -- number of equivalence queries (EQs) and membership queries (MQs) -- of the A2 algorithm.

The number of EQs and MQs asked by the algorithm is presented in Figure~\ref{figure:a2-results}.
The average number of EQs required by the A2 algorithm is constant for different number of variables and similar throughout all schemas. 
However, this is not the case for MQs. 
Particularly, we can see that the number of MQs is greater for more decomposed schemas, e.g., Original schema.
Further, the number of MQs also increases with the number of variables.



\begin{figure} 
	\includegraphics[width=0.5\linewidth]{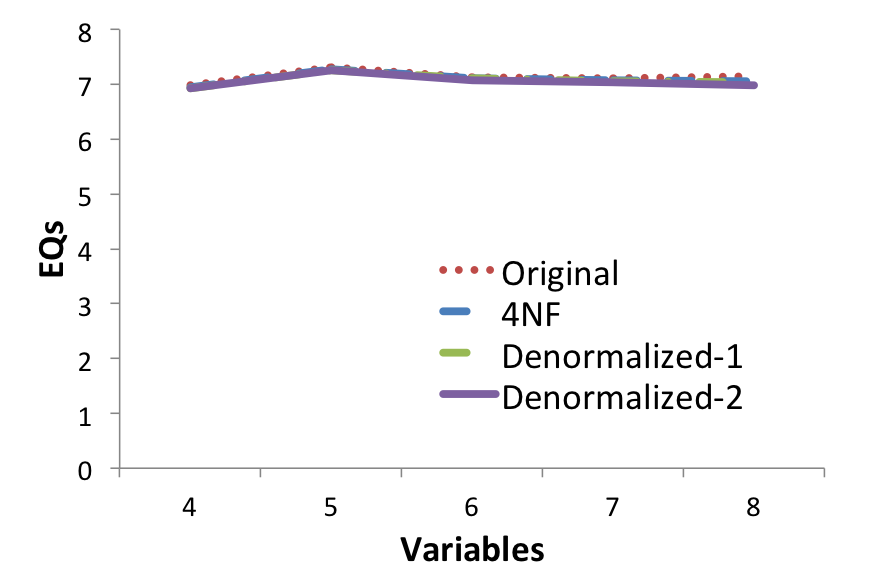}\hfill
	\includegraphics[width=0.5\linewidth]{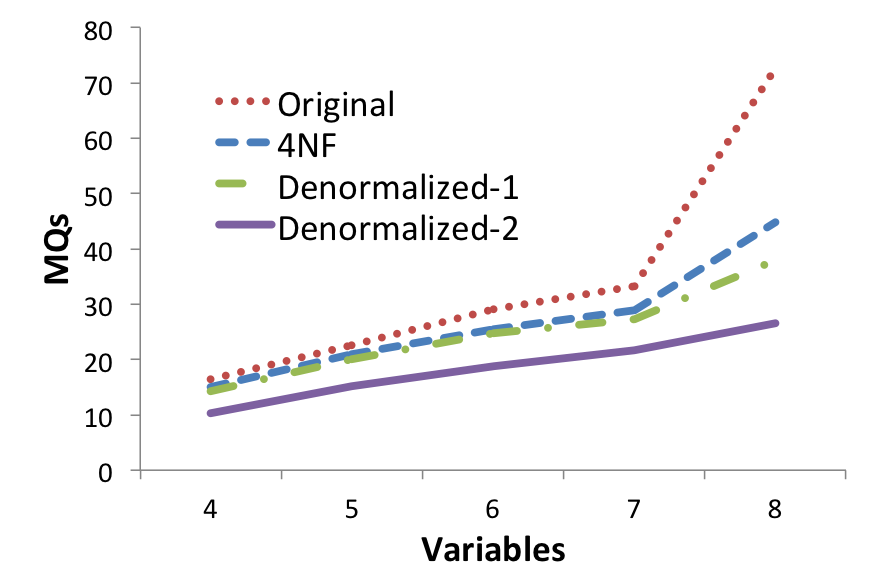}\hfill
	\caption{{\small Average number of equivalence (left) and membership (right) queries for the A2 algorithm.}}
	\label{figure:a2-results}
\end{figure}

	\section{Conclusion and Future Work}
\label{section:conclusion}
Users often have to apply relational learning algorithms over
databases that have different schemas from the ones
used to validate these algorithms. In order for these 
algorithms to effectively learn the definitions of 
target relations over various databases, 
their results should not depend on the representational details of 
the database schema. In this paper,
we formally defined the novel 
property of schema independence for
relational learning algorithms, which states that the output of
these algorithms should not depend on the schema used to represent
their input databases. 
We proved that current well-known relational 
learning algorithms are not schema independent over composition/decomposition. 
We proposed a new algorithm, Castor, that leverages schema constraints to achieve schema independence.
Our empirical results on benchmark and real datasets validated our theoretical results and showed that Castor is more efficient and more or as effective as
current relational learning algorithms.

We believe that this paper initiates some exciting 
new investigations on the impact of
representation on the quality of learning
concepts over structured data. 
An interesting direction is to explore the schema independence 
of other types of supervised and unsupervised learning algorithms over 
relational databases. 
For example, many learning systems first select and extract features from a 
relational database and then execute non-relational learning 
algorithms, such as decision trees, over these features to 
learn the target concepts \cite{RothYi01a,Anderson:CIDR:2013}.
It is also interesting to explore the 
behavior of learning algorithms over other types of 
schema transformations, such as the ones used in data exchange \cite{DBLP:conf/icdt/FaginKMP03}, 
to better understand the connection between data exchange/integration 
and learning.


	
	\bibliographystyle{ACM-Reference-Format}
	
	\bibliography{../../ref}

\end{document}